\documentclass[conference]{IEEEtran}
\IEEEoverridecommandlockouts
\usepackage{cite}
\usepackage{textcomp}
\usepackage{xcolor}
\def\BibTeX{{\rm B\kern-.05em{\sc i\kern-.025em b}\kern-.08em
    T\kern-.1667em\lower.7ex\hbox{E}\kern-.125emX}}

\usepackage{adjustbox}
\usepackage{tikz}
\usetikzlibrary{quantikz}
\usetikzlibrary{external}
\usepackage{graphicx}

\usepackage[utf8]{inputenc}
\usepackage[T1]{fontenc}
\usepackage{amsmath,amssymb,amsbsy}
\usepackage{mathtools}
\usepackage{xparse}

\newcommand{\Id}{\mathit{Id}}

\newcommand{\bigO}{\mathcal{O}}
\newcommand{\dicke}[2]{\ket{\smash{D_{#2}^{#1}}}}

\NewDocumentCommand{\bbeta}{o}{
	\pmb{\beta}\IfValueT{#1}{_{[#1]}}
}
\NewDocumentCommand{\bgamma}{o}{
	\pmb{\gamma}\IfValueT{#1}{_{[#1]}}
}
\NewDocumentCommand{\bgstate}{o}{
	\ket{	\IfValueTF{#1}{\bbeta[#1]}{\bbeta},
		\IfValueTF{#1}{\bgamma[#1]}{\bgamma}
	}
}
\NewDocumentCommand{\bgstateT}{o}{
	\bra{	\IfValueTF{#1}{\bbeta[#1]}{\bbeta},
		\IfValueTF{#1}{\bgamma[#1]}{\bgamma}
	}
}

\usepackage{amsthm}

\newtheorem{theorem}{Theorem}

\usepackage[colorlinks=true]{hyperref}

\title{Grover Mixers for QAOA: Shifting Complexity from Mixer Design to State Preparation%
\thanks{Research presented in this article was supported by the Center for Nonlinear Studies CNLS and the Laboratory Directed Research and Development program of Los Alamos National Laboratory under project numbers 20190495CR and 20200671DI.
\hfill LA-UR-20-23893
}}
\author{
	\IEEEauthorblockN{Andreas Bärtschi}
	\IEEEauthorblockA{\textit{CCS-3 Information Sciences and CNLS} \\
	\textit{Los Alamos National Laboratory}\\
	Los Alamos, NM 87544, USA \\
	baertschi@lanl.gov}
\and
	\IEEEauthorblockN{Stephan Eidenbenz}
	\IEEEauthorblockA{\textit{CCS-3 Information Sciences} \\
	\textit{Los Alamos National Laboratory}\\
	Los Alamos, NM 87544, USA \\
	eidenben@lanl.gov}
}

\begin{document}
\maketitle

\begin{abstract}
	We propose GM-QAOA, a variation of the Quantum Alternating Operator Ansatz (QAOA) that uses Grover-like selective phase shift mixing operators. 
	GM-QAOA works on any NP optimization problem for which it is possible to efficiently prepare an equal superposition of all feasible solutions;
	it is designed to perform particularly well for \emph{constraint} optimization problems, where not all possible variable assignments are feasible solutions. 
	GM-QAOA has the following features: (i) It is not susceptible to Hamiltonian Simulation error (such as Trotterization errors) as its operators can be implemented exactly using standard gate sets and (ii) Solutions with the same objective value are always sampled with the same amplitude. 

	We illustrate the potential of GM-QAOA on several optimization problem classes: 
	for permutation-based optimization problems such as the Traveling Salesperson Problem, we present an efficient algorithm to prepare a superposition of all possible permutations of $n$ numbers, defined on $O(n^2)$ qubits; 	
	for the hard constraint $k$-Vertex-Cover problem, and for an application to Discrete Portfolio Rebalancing, we show that GM-QAOA outperforms existing QAOA approaches.
\end{abstract}

\section{Introduction}

Combinatorial optimization is among the key applications to benefit from error-tolerant quantum computing as it could speed up practically and theoretically relevant problems such as \textsc{Minimum Traveling Salesperson} or \textsc{Maximum Satisfability}. While we generally do not expect to find polynomial-time quantum algorithms for such $NP$-hard optimization problems, quantum optimization algorithms that improve runtimes and/or quality (in terms of approximation ratios of found solutions vs. optimum solutions) over their classical counter parts may still exist. 

In this paper, we propose the polynomial-time Grover-Mixer Quantum Alternating Operator Ansatz (GM-QAOA) algorithm as a combination of the two leading quantum optimization approaches, namely Grover-based and Quantum Alternating Operator Ansatz (QAOA)-based~\cite{Hadfield2019}. Typically, quantum optimization algorithms that find optimum solutions at the cost of exponential running times with reduced base (usually $>2$) over their classical counter parts, such as \cite{ambainis2019quantum}, use Grover search as a key component. Quantum approximation algorithms that aim to find good approximate solutions to optimization problems in polynomial time typically follow the Quantum Alternating Operator Ansatz (QAOA) \cite{Farhi2014,Hadfield2019}.

\begin{figure*}[t]
	\centering
	\begin{adjustbox}{width=0.8\textwidth}
	\newcommand{\zgate}[1]{\gate{Z^{-\beta_{#1}/\pi}}}	
	\begin{quantikz}[row sep={24pt,between origins},execute at end picture={
				\node[xshift=-6pt, yshift=-20pt] at (\tikzcdmatrixname-5-2) {
					\smash{\textcolor{red}{$\frac{1}{\sqrt{|F|}}\sum\limits_{x\in F} \ket{x}$}}
				};	
				\node[xshift=-4pt, yshift=-24pt] at (\tikzcdmatrixname-5-6) {
					$\underbrace{\hspace*{260pt}}_{p\text{ rounds with angles } \gamma_1,\beta_1,\ldots,\gamma_p,\beta_p}$
				};
			}]
		\lstick{\ket{0}}	
		& \gate[5]{U_S}\slice{}
		& \gate[5]{U_P(\gamma_k)}
		& \gate[5]{U_S^{\dagger}}
		\gategroup[5,steps=5,style={dashed,rounded corners,fill=blue!20, inner sep=0pt},background]{$U_M(\beta_k) = e^{-i\beta_k \ket{F}\bra{F}}$}
		& \targ{}
		& \ctrl{1}
		& \targ{}
		& \gate[5]{U_S}
		& \qw\midstick[5,brackets=none]{\ldots}
		& \meter{}\rstick[5]{\rotatebox{90}{$\bgstateT H_C \bgstate$}}
		\\	
		\lstick{\ket{0}}	&&&	& \targ{} 	& \ctrl{1}	& \targ{} 	&	&\qw	& \meter{}		\\	
		\lstick{\ket{0}}	&&&	& \targ{}	& \ctrl{1}	& \targ{}	&	&\qw	& \meter{}		\\	
		\lstick{\ket{0}}	&&&	& \targ{}	& \ctrl{1}	& \targ{}	&	&\qw	& \meter{}		\\	
		\lstick{\ket{0}}	&&&	& \targ{}	& \zgate{k}	& \targ{}	&	&\qw	& \meter{}	
	\end{quantikz}
	\end{adjustbox}
	\caption{GM-QAOA (Grover Mixer QAOA): The state preparation unitary $U_S$ preparing the 
		equal superposition of all feasible states $\ket{F} = \frac{1}{\sqrt{|F|}}\sum_{x\in F} \ket{x}$
		and its conjugate transpose $U_S^{\dagger}$
		are used to implement the Mixer $U_M(\beta) = e^{-i\beta \ket{F}\bra{F}}$.
	}
	\label{fig:gmqaoa}
\end{figure*}
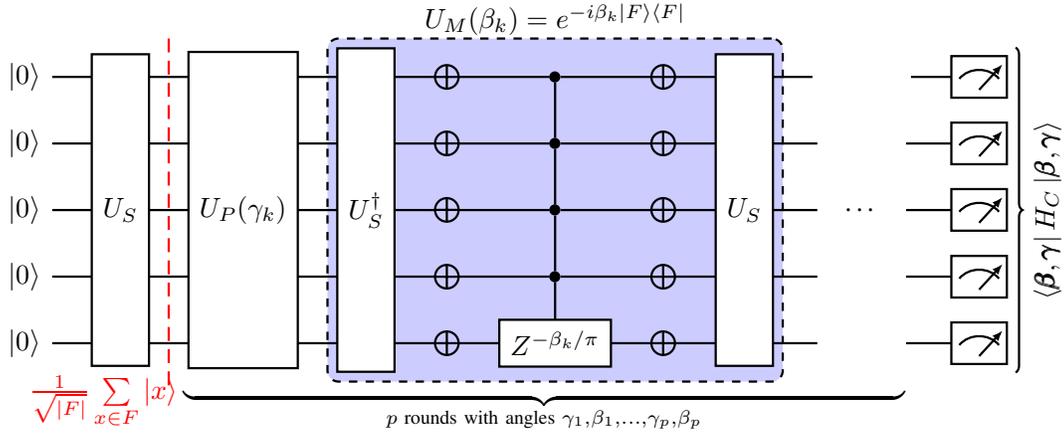

\subsection{Overview of GM-QAOA: Grover-Mixer Quantum Alternating Operator Ansatz}
\label{sec:gm-qaoa}
In order to introduce GM-QAOA more formally, we need a bit of notation. We use the combinatorial optimzation problem of \textsc{Max k-Vertex Cover} as an illustrative example.
An instance $I^{kvc}$ of \textsc{Max k-Vertex Cover} is a tuple $(G, k)$ of a graph $G =(V,E)$ with vertex set $V = \{v_0, \ldots, v_{n-1}\}$ and edge set $E \subseteq V^2 = \{\forall 0 \leq i,j < n: (v_i, v_j)\}$ and 
we are to find a set $V' \subset V$ of exactly $k$ vertices, such that a maximum number of edges in $E$ have at least one end point in $V'$. I.e, the cost function is 
\begin{align}
    C(V') = \sum_{(v_i, v_j) \in E \text{ with } v_i \in V' \text{ or } v_j \in V'} 1 \label{eq:kvccost}
\end{align}

Just like standard QAOA, GM-QAOA takes as input a triple $(I, F, C)$ of optimization problem instance $I$, feasible solutions $F$ and cost function $C$ as follows:
\begin{itemize}
    \item 
 $I$ is an instance of an optimization problem, whose solutions are defined on $n$ binary variables $x_0, \ldots, x_{n-1}$. In our \textsc{Max k-Vertex Cover} instance $I^{kvc}$, the binary variables $x_i$ encode whether vertex $v_i$ is in a vertex cover solution $V'$ or not.
\item
Set $F$ is the set of feasible solutions, where a solution $x = x_{n-1}\ldots x_0 \in F \subseteq \left\{ 0,1 \right\}^n$ is the solution with variables taking values $x_i$. If  $F \subsetneq \left\{ 0,1 \right\}^n$, we call the optimization problem a \textit{constraint} optimization problem. In our \textsc{Max k-Vertex Cover} example of instance $I^{kvc}$, the set $F$ of feasible solutions consists of all binary strings  $x$ of length $n$ with exactly $k$ variables set to 1, or alternatively speaking,  all length-$n$ strings of Hamming weight $k$. 
\item
Cost function $C$ assigns a cost value to each feasible solution $x$ as given above for instance $I^{kvc}$, where $C(x)$ is the number of edges with at least one end point vertex with corresponding $x_i =1 $ in solution $x$. The transformation $x_i \mapsto (1-\sigma_i^z)/2$ yields the cost Hamiltonian
$$ H_C = C((1-\sigma_0^z)/2, \ldots, (1-\sigma_{n-1}^z)),
$$
and a solution $x$ has objective value $\bra{x}H_C\ket{x}$.
\end{itemize}

As illustrated in Fig.~\ref{fig:gmqaoa}, the GM-QAOA circuit of level $p$ for input $(I,C, F)$ consists of
\begin{enumerate}
	\item a	 \emph{state preparation} unitary operator $U_S$ that creates (without measurement) an equal superposition of all feasible solutions in $F$:
	$$
	    U_S \ket{0} = \ket{F} := \frac{1}{\sqrt{|F|}}\sum\limits_{x\in F} \ket{x}
	$$
	\item	$p$ applications of alternating \emph{phase separation} and \emph{mixing} unitaries $U_P$, $U_M$,
	\item	and a final measurement in the computational basis. 	
\end{enumerate}
The phase separator and mixing unitaries are parameterized each by $p$ real numbers (angles) $\bgamma = (\gamma_1, \ldots, \gamma_p)^T$ and $\bbeta = (\beta_1, \ldots, \beta_p)^T$; we write the final state before measurement as 
\[	\bgstate := U_M(\beta_p) U_P(\gamma_p) \cdots U_M(\beta_1) U_P(\gamma_1) U_S \ket{0}^{\otimes n}.	\]

The phase separator unitary $U_P$ should be diagonal in the computational basis. We use $U_P(\gamma) = e^{-i\gamma H_C}$ (modulo global phases), where $H_C$ is the cost Hamiltonian. As in standard QAOA, the angle vectors $\bbeta, \bgamma$ need to be optimized over in an outer-loop, thus GM-QAOA remains a variational algorithm; finding optimum or near-optimum angles can be done through a variety of optimization techniques \cite{nasa,eidenbenz2019quantum} or even theoretical analysis of individual values \cite{Farhi2015}.

The key innovation of GM-QAOA lies in the mixing unitary $U_M$ as we focus on state preparation to prepare an equal superposition of all feasible states and design a mixing unitary based thereon. 
In particular, once we have an efficient state preparation unitary $U_S$ such that $U_S\ket{0} = \ket{F} := |F|^{-1/2} \sum_{x\in F}\ket{x}$,
we can directly implement a variational (with parameter $\beta$) mixing unitary
\begin{align}
	U_M(\beta)	
	&= e^{-i\beta\ket{F}\bra{F}}
\end{align}
with
\begin{align}
	e^{-i\beta\ket{F}\bra{F}}
	&= \sum_{k=0}^{\infty} \frac{(-i\beta)^k (\ket{F}\bra{F})^k}{k!}    \nonumber \\
	&= \Id + \sum_{k=1}^{\infty} \frac{(-i\beta)^k}{k!}\ket{F}(\langle F|F\rangle)^{k-1}\bra{F} \nonumber\\
	&=	\Id - \left( 1-e^{-i\beta} \right)\ket{F}\bra{F} \nonumber\\
		&= 	U_S \left(\Id - \left( 1-e^{-i\beta} \right) \ket{0}\bra{0}\right) U_S^{\dagger}.	\label{eq:grover-mixer}
\end{align}

Such mixing unitaries resemble the well-known diffusion operators used in amplitude amplification with the phase shift $e^{-i\beta}$ replacing phase inversion (i.e. phase shifts of $-1 = e^{-i\pi}$). 
Thus the mixing unitary can be implemented using the one application of unitaries $U_S, U_S^{\dagger}$ each, two layers of $X$-gates and a multi-control phase-shift gate $Z^{-\beta/\pi} = \left(\begin{smallmatrix} 1 & 0 \\ 0 & e^{-i\beta} \end{smallmatrix}\right)$, which can be implemented in linear depth and size, as shown in Section~\ref{sec:proof}.
In fact, such selective phase shift unitaries have been introduced by Grover with angles $\beta = \pm \pi/3$ for fixed-point quantum search \cite{grover_fixed_point} and were later generalized to a larger discrete set of angles by Yoder et al. \cite{yoder2014fixed}. We discuss the key properties of GM-QAOA that follow from this definition.

\subsubsection{Property 1: Computing Equal Feasible State Superpositions}
If we limit ourselves to polynomial-size circuits, GM-QAOA relies on the existence of a polynomial state preparation algorithm for unitary $U_S$ to create an equal superposition of all feasible solutions $F$ of input instance $I$. It is not at all obvious that such quantum algorithms exist for all combinatorial optimization problems. At first glance, three categories of optimization problems exist with respect to the feasible solution space $F$:
\begin{enumerate}
    \item All solutions are feasible (i.e, $F = \left\{ 0,1 \right\}^n$, we call the problem unconstrained). In the original Quantum Approximate Optimization Algorithm (QAOA), Farhi et. al. studied unconstrained binary optimization problems such as \textsc{MaxCut}~\cite{Farhi2014} and \textsc{MAxE3Lin2}~\cite{Farhi2015}. We trivially get an equal superposition of all computational basis states $\ket{+}^{\otimes n} = H^{\otimes n} \ket{0}^{\otimes n}$ in circuit depth 1 (Fig.~\ref{fig:qaoa1}).
    \item The optimization problem is constrained (i.e., $F \subset \left\{ 0,1 \right\}^n$) and there exists a polynomial-sized quantum circuit to implement $U_S$. For example, Dicke-states (superpositions of all equal Hamming-weight states) can be computed in polynomial time \cite{dickestates} and serve as $U_S$ for a large set of optimization problems, such as \textsc{Max k-Vertex Cover} or \textsc{Max k-Set Cover}. 
    \item No polynomial-time quantum algorithm is known to compute $U_S$. We conjecture that problems such as \textsc{Max Clique} belong to this category.
\end{enumerate}

GM-QAOA performs best on problems of Category 2, which will be the focus of this paper, but can also handle problems in Category 1 (albeit at larger depth of the mixing unitary). As long as we stick to polynomial-size circuits, GM-QAOA will not handle Category-3 problems.

\subsubsection{Property 2: Mixing equal solutions at equal amplitude}
Our Grover mixer unitaries are such that GM-QAOA results in final states $\bgstate$ where any two basis states $x,y \in F$ with the same objective value, that is $C(x) = C(y)$ or equivalently 
$\bra{x}H_C\ket{x} = \bra{y}H_C\ket{y}$, have the same amplitude in $\bgstate$. To the best of our knowledge, this is the first mixing operator that has this property (even including Category 1). We formally prove this result in Section \ref{sec:proof}.

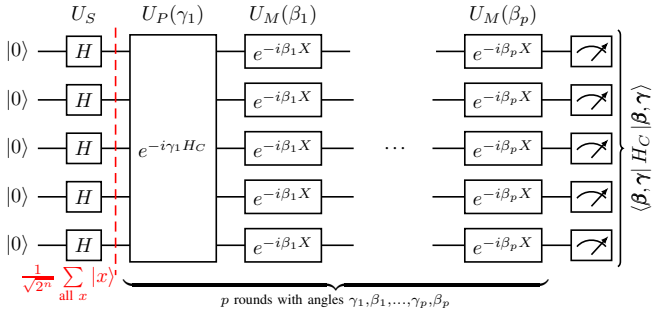
\begin{figure}[t!]
	\centering
	\newcommand{\rxgate}[1]{\gate{e^{-i\beta_{#1}X}}}
	\begin{adjustbox}{width=\linewidth}
	\begin{quantikz}[row sep={24pt,between origins},execute at end picture={
				\node[xshift=-8pt, yshift=-18pt] at (\tikzcdmatrixname-5-2) {
					\smash{\textcolor{red}{$\frac{1}{\sqrt{2^n}}\sum\limits_{\text{all }x} \ket{x}$}}
				};
				\node[yshift=18pt] at (\tikzcdmatrixname-1-2) {$U_S$};
				\node[yshift=18pt] at (\tikzcdmatrixname-1-3) {$U_P(\gamma_1)$};
				\node[yshift=18pt] at (\tikzcdmatrixname-1-4) {$U_M(\beta_1)$};
				\node[yshift=18pt] at (\tikzcdmatrixname-1-6) {$U_M(\beta_p)$};
				\node[xshift=26pt, yshift=-24pt] at (\tikzcdmatrixname-5-4) {
					$\underbrace{\hspace*{210pt}}_{p\text{ rounds with angles } \gamma_1,\beta_1,\ldots,\gamma_p,\beta_p}$
				};
			}]
		\lstick{\ket{0}}	
		& \gate{H}\slice{}	
		& \gate[5]{e^{-i\gamma_1 H_C}}
		& \rxgate{1}
		& \qw\midstick[5,brackets=none]{\ \ldots\ }
		& \rxgate{p}
		& \meter{}\rstick[5]{\rotatebox{90}{$\bgstateT H_C \bgstate$}}
		\\	
		\lstick{\ket{0}}	& \gate{H}	&	& \rxgate{1}	&\qw	& \rxgate{p}	& \meter{}		\\	
		\lstick{\ket{0}}	& \gate{H}	&	& \rxgate{1}	&\qw	& \rxgate{p}	& \meter{}		\\	
		\lstick{\ket{0}}	& \gate{H}	&	& \rxgate{1}	&\qw	& \rxgate{p}	& \meter{}		\\	
		\lstick{\ket{0}}	& \gate{H}	&	& \rxgate{1}	&\qw	& \rxgate{p}	& \meter{}	
	\end{quantikz}
	\end{adjustbox}
	\caption{Original Quantum Approximate Optimization Algorithm for unconstrained problems. 
	State Preparation and a Transverse Field Hamiltonian Mixer can be implemented with single-qubit gates in depth 1.}
	\label{fig:qaoa1}
\end{figure}

\subsubsection{Property 3: No Hamiltonian Simulation error}
The GM-QAOA algorithm can be analyzed and understood purely in terms of Grover-like paradigms and elementary Pauli gates, which  -- as an aside -- implies that the reader does not need a background in Hamiltonian simulation principles in order to understand GM-QAOA. 
It follows as a benefit  that GM-QAOA does not suffer from numerical errors that Hamiltonian simulation methods, such as Trotterization~\cite{lloyd1996trotterization}, Local Combination of Unitaries~\cite{berry2015simulating}, or Quantum Walk based simulation~\cite{berry2015hamiltonian}, introduce. 
This in turn decreases the circuit complexity of GM-QAOA, making it more likely to succeed in a noisy environment, even though it may not particularly NISQ friendly.

\subsection{GM-QAOA vs. standard QAOA}
In order to appreciate GM-QAOA's standing among quantum optimization algorithms, we contrast it to earlier work. 
In the original Quantum Approximate Optimization Algorithm, Farhi et. al. studied \emph{unconstrained} binary optimization problems, for which any solution is feasible and we get an equal superposition of all computational basis states $\ket{+}^{\otimes n} = H^{\otimes n} \ket{0}^{\otimes n}$ in circuit depth 1. Fig.~\ref{fig:qaoa1} shows an overview of this original QAOA approach.

Hadfield et. al. extended QAOA into a general framework~\cite{Hadfield2019}, renamed to Quantum Alternating Operator Ansatz,  to cover a wide range of combinatorial optimization problems, including constraint problems.  
Fig.~\ref{fig:qaoa2} shows an overview of the Hadfield QAOA approach. 
Unlike GM-QAOA, Hadfield QAOA recommends for state preparation that $U_S$ should efficiently create a superposition of \emph{some feasible states} in $F$. 
Hadfield defines several classes of mixing unitaries that are suited to different classes of optimization problems. In keeping with the tradition of classical discrete optimization, the QAOA ansatz requests that mixing unitaries $U_M$ preserve the feasible subspace and provide transitions between all pairs of feasible states, i.e. for every feasible state $y\in F$ we have:	``there exists an angle $\beta^*$ such that $\lvert\bra{x} U_M(\beta^*) \ket{y}\rvert > 0$'' if and only if ``$x\in F$''.

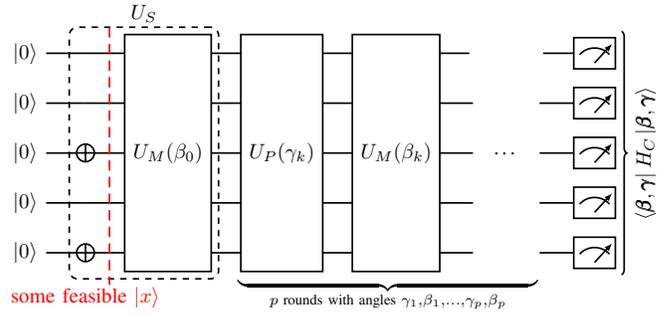
\begin{figure}[t]
	\centering
	\begin{adjustbox}{width=\linewidth}
	\begin{quantikz}[row sep={24pt,between origins},execute at end picture={
				\node[yshift=-24pt] at (\tikzcdmatrixname-5-2) {
					\smash{\textcolor{red}{some feasible $\ket{x}$}}
				};
				\node[xshift=-4pt, yshift=-20pt] at (\tikzcdmatrixname-5-5) {
					$\underbrace{\hspace*{144pt}}_{p\text{ rounds with angles } \gamma_1,\beta_1,\ldots,\gamma_p,\beta_p}$
				};
			}]
		\lstick{\ket{0}}	
		& \qw\slice{}\gategroup[5,steps=2,style={dashed,rounded corners, inner sep=0pt},background]{$U_S$}
		& \gate[5]{U_M(\beta_0)}
		& \gate[5]{U_P(\gamma_k)}
		& \gate[5]{U_M(\beta_k)}
		& \qw\midstick[5,brackets=none]{\ldots}
		& \meter{}\rstick[5]{\rotatebox{90}{$\bgstateT H_C \bgstate$}}
		\\	
		\lstick{\ket{0}}	& \qw		&	&	&	&\qw	& \meter{}		\\	
		\lstick{\ket{0}}	& \targ{}	&	&	&	&\qw	& \meter{}		\\	
		\lstick{\ket{0}}	& \qw		&	&	&	&\qw	& \meter{}		\\	
		\lstick{\ket{0}}	& \targ{}	&	&	&	&\qw	& \meter{}	
	\end{quantikz}
	\end{adjustbox}
	\caption{Quantum Alternating Operator Ansatz for Constraint Optimization Problems with a focus on Mixer Design. 
		An intial Mixer application -- applied to an easily prepared single feasible state -- is also used for State Preparation.}
	\label{fig:qaoa2}
\end{figure}

As for initial state preparation, the standard approach by Hadfield is to start in a \emph{ single} feasible computational basis state (implementable in depth 1 with Pauli $X$ gates) and to apply an initial mixing unitary $U_M(\beta_0)$. In this case, the efficiency of the initial state preparation corresponds to the efficiency of the mixer implementation. 
In particular, it is highly non-trivial to create mixer unitaries that assign amplitude to all feasible states starting from just a single feasible basis state; in practice, the proposed QAOA mixers do not achieve such perfect mixing in a single iteration, particularly when restricted to polynomial circuit size, thus the mixing across all feasible states happens only after several iterations of mixer and phase separator, which in turn makes the entire QAOA circuit longer and thus less NISQ-friendly. 

Numerical simulations for applications with $XY$-model mixers such as \textsc{Max-k-colorable Subgraph} and \textsc{Max-k-VertexCover} found
that starting in an equal superposition of all feasible states (e.g. $W$- and Dicke-States for the mentioned examples) offers a benefit over a random computational basis state choice followed by an angle-optimized initial mixing unitary $U_M(\beta_0)$~\cite{nasa,eidenbenz2019quantum}. 

In a sense, our GM-QAOA follows a reverse approach:
Focus on state preparation to prepare an equal superposition of all feasible states and make use of Grover's selective phase shift operator to get a relatively simple mixing unitary that always remains and provides good transition properties in the feasible solution space $F$. We thus shift the complexity from the mixing unitary into the preparation of $U_S$. This works well and brings a real benefit only if $U_S$ (and $U_S^\dagger$) can be constructed in polynomial time.
Hadfield's QAOA ansatz is general enough that GM-QAOA could be considered a special case of the ansatz for such problems.

\paragraph*{Related Work}
The original Grover quantum search algorithm~\cite{grover1996} with phase inversions $e^{-i\pi \ket{t}\bra{t}}, e^{-i\pi \ket{+}\bra{+}^{\otimes n}}$ around a target state $\ket{t}$ and an initial state $\ket{+}^{\otimes n}$ has a ``Soufflé'' problem~\cite{Brassard1997}, in which too few or two many iterations ``undercook'' (or ``overcook'', respectively) the amplitude of the target state, leading to a sampling of mostly unmarked states. The same holds true for its generalization of amplitude amplification~\cite{brassard2000}, which combines the state preparation and diffusion operator of Grover's quantum search with state preparation unitaries $U_S$ that correspond to, e.g., a classical randomized algorithm or a restriction of the Hilbert space to a feasible subspace. 
This was addressed by Grover himself with a fixed-point quantum search algorithm~\cite{grover_fixed_point} that replaced selective phase inversions with selective phase shifts of $\pm \pi/3$. This algorithm converges to the target state $\ket{t}$, albeit at a loss of the quadratic speedup. Using a larger set of angles, Yoder et~al.~\cite{yoder2014fixed} later recovered the quadratic speedup by replacing exact convergence with convergence up to a tunable error parameter.

Variational learning of phases for quantum search~\cite{biamonte2018} and Grover mixers for QAOA have recently been studied for \emph{unconstrained} problems, such as for \textsc{Max-SAT}~\cite{biamonte2020} and for counting of weighted ground states~\cite{sundar2019}. The Hamiltonian 
$\ket{+}\bra{+}^{\otimes n} \cong \Id + \sum_i \sigma_i^x + \sum_{\langle i,j \rangle} \sigma_i^x \sigma_j^x + \sum_{\langle i,j,k \rangle} \sigma_i^x \sigma_j^x \sigma_k^x \ldots$ has also been proposed for fair sampling in quantum annealing~\cite{Matsuda2009}.

\paragraph*{Outline}
In the remainder of the paper, we first formally show the properties of the GM-QAOA mixers in Section~\ref{sec:proof}.
The task of creating non-trivial starting states, which are typically highly entangled states, is of independent interest as well.
In particular, they can also be used for efficient implementations 
of starting states and oracle operators in Grover Search algorithms for these problems.
We illustrate the potential of our method on various problems, such as permutation-based optimization such as the Traveling Salesperson Problem,
optimization with hard constraints such as $k$-Vertex-Cover, and an application to portfolio rebalancing in Section~\ref{sec:app}. 
We describe efficient state preparation circuits and compare the resulting mixers with current proposals for mixing unitaries for these problems.

\section{GM-QAOA: Proofs and Implementation}
\label{sec:proof}

For a vector $\pmb{v}$, denote by $\pmb{v}_i$ the $i$-th component of $\pmb{v}$.
In particular, we denote by 
$\bgstate_i$ 
the amplitude of the $i$-th computational basis state.

\begin{theorem}
	Consider any GM-QAOA circuit with a state preparation unitary $U_S\colon \ket{0}^{\otimes n} \mapsto \ket{F}$, phase separators $U_P(\gamma) = e^{-i\gamma H_C}$, and Grover mixers $U_M(\beta) = e^{-i\beta\ket{F}\bra{F}} = U_S (\Id-(1-e^{-i\beta})\ket{0}\bra{0}) U_S^{\dagger}$. 
	Then we have final amplitudes
	\begin{description}
		\item[$\bgstate_z = 0$] \qquad\qquad\ if $z \notin F$ (staying in the feasible subspace),
		\item[$\bgstate_x = \bgstate_y$] \qquad\qquad\ whenever $x,y \in F$ have the same objec-\newline\hspace*{40pt}
		tive value $\bra{x}H_C\ket{x} = \bra{y}H_C\ket{y}$.
	\end{description}
\end{theorem}

\newcommand{\obj}{\mathrm{obj}}
\newcommand{\AM}{\mathit{AM}}
\begin{proof}
	We proceed by induction over the number of rounds~$k$, denoting by $\bgstate[k]$ the state after the first $k$ QAOA rounds.
	For feasible states $f\in F$ define $\obj(f) := \bra{f}H_C\ket{f}$.
	
	\paragraph*{State preparation} ($k=0$):
	We have $\bgstate[0] = \ket{F}$, and every feasible state $f \in F$ has amplitude $\bgstate[0]_z = \smash{\tfrac{1}{\sqrt{|F|}}}$, 
	while all other states $z\notin F$ have amplitude 0.
	
	\paragraph*{QAOA rounds} ($p \geq k > 0$): 
	Assume by induction hypothesis that for $x,y \in F$ with $\obj(x)=\obj(y)$ we have $\bgstate[k-1]_x = \bgstate[k-1]_y$,
	and that for $z \notin F$, $\bgstate[k-1]_z=0$.
	We show $\bgstate[k]_x = \bgstate[k]_y$ and $\bgstate[k]_z=0$.

	Since $U_P(\gamma_{k})$ is diagonal in the computational basis, it phase shifts feasible states $f\in F$ by $e^{-i\gamma_{k} \obj(f)}$ 
	and we get $( U_P(\gamma_{k})\bgstate[k] )_f = e^{-i\gamma_{k} \obj(z)} \bgstate[k-1]_f$,
	while staying in the feasible subspace. 
	Thus the amplitudes of the new state have an arithmetic mean of 
	\begin{align*}
		\AM &:=	\AM\left( U_P(\gamma_{k})\bgstate[k]\right)	\\
		&=	\frac{1}{|F|} \sum_{f \in F} e^{-i\gamma_{k} \obj(f)}\bgstate[k]_f	\\		
		&=	\sqrt{|F|} \cdot \bra{F} U_P(\gamma_{k})\bgstate[k].
	\end{align*}
	After application of $U_M(\beta_{k}) = \left(\Id - (1-e^{-i\beta_{k}})\ket{F}\bra{F}\right)$, we stay in the feasible subspace $F$:
	\begin{align*}
		\bgstate[k]	&=	U_M(\beta_{k}) U_P(\gamma_{k})\bgstate[k-1]\\	
		\		&=	U_P(\gamma_{k})\bgstate[k-1] - \frac{(1-e^{-i\beta_{k}})}{\sqrt{|F|}} \AM \ket{F}	
	\end{align*}
	Thus we have for $x,y\in F$ with $\obj(x) = \obj(y)$ that
	\begin{align*}
		\bgstate[k]_x	&=	e^{-i\gamma_{k}\obj(x)}\bgstate[k-1]_x - \tfrac{(1-e^{-i\beta_{k}})}{\sqrt{|F|}} \AM	\\
		\		&=	e^{-i\gamma_{k}\obj(y)}\bgstate[k-1]_y - \tfrac{(1-e^{-i\beta_{k}})}{\sqrt{|F|}} \AM	\\
		\		&=	\bgstate[k]_y.\qedhere   
	\end{align*}
\end{proof}

\subsubsection*{Implementation of GM-QAOA phase separator and mixing unitaries}
Here we briefly review implementations for phase separators $U_P(\gamma) = e^{-i\gamma H_C}$,
and for multi-controlled phase-shift operators $Z^{t} = \left(\begin{smallmatrix} 1 & 0 \\ 0 & e^{i t*pi} \end{smallmatrix}\right)$ used in GM-QAOA's mixing unitaries. 

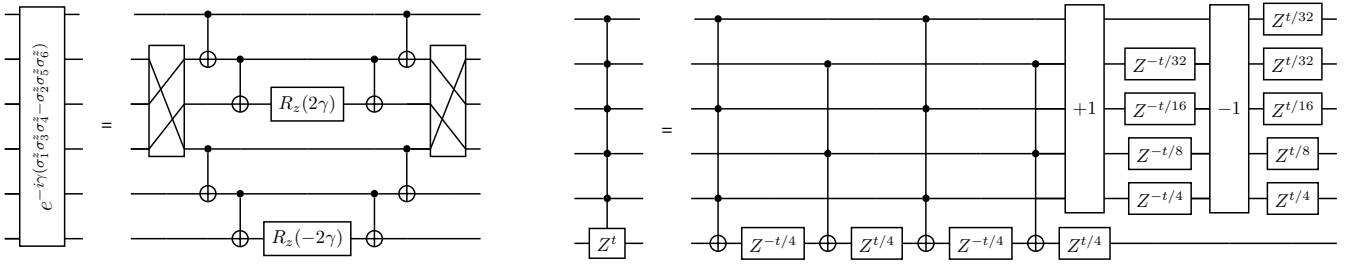
\begin{figure*}[th]
	\centering
	\begin{adjustbox}{width=\linewidth}
	\def\rowspace{24pt}
	\newcommand{\rzgate}[1]{\gate{R_z(#1 2\gamma)}}
	\begin{quantikz}[row sep={24pt,between origins},column sep=8pt,execute at end picture={
				\draw[thick] ($(\tikzcdmatrixname-2-4.east)-(0pt,0pt)$) to ($(\tikzcdmatrixname-2-4.west)-(0pt,\rowspace)$);
				\draw[thick] ($(\tikzcdmatrixname-2-4.east)-(0pt,\rowspace)$) to ($(\tikzcdmatrixname-2-4.west)-(0pt,2*\rowspace)$);
				\draw[thick] ($(\tikzcdmatrixname-2-4.east)-(0pt,2*\rowspace)$) to ($(\tikzcdmatrixname-2-4.west)-(0pt,0pt)$);	
				\draw[thick] ($(\tikzcdmatrixname-2-10.west)-(0pt,0pt)$) to ($(\tikzcdmatrixname-2-10.east)-(0pt,\rowspace)$);
				\draw[thick] ($(\tikzcdmatrixname-2-10.west)-(0pt,\rowspace)$) to ($(\tikzcdmatrixname-2-10.east)-(0pt,2*\rowspace)$);
				\draw[thick] ($(\tikzcdmatrixname-2-10.west)-(0pt,2*\rowspace)$) to ($(\tikzcdmatrixname-2-10.east)-(0pt,0pt)$);
			}]
		& \gate[6][24pt]{{\rotatebox{90}{\large $\mathclap{e^{-i\gamma (\sigma^z_1 \sigma^z_3 \sigma^z_4 - \sigma^z_2 \sigma^z_5 \sigma^z_6)}}$}}}
		& \midstick[6,brackets=none]{=}\qw& \qw		& \ctrl{1}	& \qw		& \qw		& \qw		& \ctrl{1}	& \qw		& \qw	\\
		& \qw		& \qw		& \gate[3]{}W	& \targ{}	& \ctrl{1}	& \qw		& \ctrl{1}	& \targ{}	& \gate[3]{}W	& \qw	\\
		& \qw		& \qw		& \qw		& \qw		& \targ{}	& \rzgate{}	& \targ{}	& \qw		& \qw		& \qw	\\
		& \qw		& \qw		& \qw		& \ctrl{1}	& \qw		& \qw		& \qw		& \ctrl{1}	& \qw		& \qw	\\
		& \qw		& \qw		& \qw		& \targ{}	& \ctrl{1}	& \qw		& \ctrl{1}	& \targ{}	& \qw		& \qw	\\
		& \qw		& \qw		& \qw		& \qw		& \targ{}	& \rzgate{-}& \targ{}	& \qw		& \qw		& \qw	
	\end{quantikz}
	\qquad\qquad
	\begin{quantikz}[row sep={24pt,between origins},column sep=8pt]
		& \ctrl{5}	& \midstick[6,brackets=none]{=}\qw
		& 				  \ctrl{5}	& \qw		& \qw		& \qw		& \ctrl{5}	& \qw		& \qw		& \gate[5]{+1}	& \qw			& \gate[5]{-1}	& \gate{Z^{t/32}}	& \qw		\\
		& \ctrl{}	& \qw		& \qw		& \qw		& \ctrl{4}	& \qw		& \qw		& \qw		& \ctrl{4}	& \qw		& \gate{Z^{-t/32}}	& \qw		& \gate{Z^{t/32}}	& \qw		\\
		& \ctrl{}	& \qw		& \ctrl{}	& \qw		& \qw		& \qw		& \ctrl{}	& \qw		& \qw		& \qw		& \gate{Z^{-t/16}}	& \qw		& \gate{Z^{t/16}}	& \qw		\\
		& \ctrl{}	& \qw		& \qw		& \qw		& \ctrl{}	& \qw		& \qw		& \qw		& \ctrl{}	& \qw		& \gate{Z^{-t/8}}	& \qw		& \gate{Z^{t/8}}	& \qw		\\
		& \ctrl{}	& \qw		& \ctrl{}	& \qw		& \qw		& \qw		& \ctrl{}	& \qw		& \qw		& \qw		& \gate{Z^{-t/4}}	& \qw		& \gate{Z^{t/4}}	& \qw		\\
		& \gate{Z^t}	& \qw		& \targ{}	& \gate{Z^{-t/4}}& \targ{}	& \gate{Z^{t/4}}& \targ{}	& \gate{Z^{-t/4}}& \targ{}	& \gate{Z^{t/4}}	& \qw			& \qw		& \qw			& \qw		
	\end{quantikz}
	\end{adjustbox}
	\caption{(left) Implementation of a phase separator $U_P(\gamma) = e^{-i\gamma H_C}$ for a cost Hamiltonian $H_C = \sigma^z_1\sigma^z_3\sigma^z_4 - \sigma^z_2\sigma^z_5\sigma^z_6$ using a swap network, parity computing CNOT stairs and Z-Rotations $R_z(2\gamma)$.
	(right) Overview of a compilation of a multi-controlled phase shift gate $Z^t$ to into single-qubit phase shift gates and a constant number of Multi-Toffoli and Increment gates, both of which can be further compiled into a linear number of 1- and 2-qubit gates. 
	}
	\label{fig:compilation}
\end{figure*}

The cost Hamiltonians $H_C$ arising in combinatorial optimization problems are diagonal 
in the computational basis and can be written as sums of local Pauli-Z products. For example, 
the objective function of \textsc{MaxE3Lin2} is a sum of 3-local terms 
$H_C = \sum d_{uvw} \sigma^z_u \sigma^z_v \sigma^z_w$ with $d_{uvw} = \pm1$.
Since Pauli-Z terms pairwise commute, we can implement each term individually and in arbitrary order,
\[ U_P(\gamma) = e^{-i\gamma \sum d_{uvw} \sigma^z_u \sigma^z_v \sigma^z_w} 
    = \prod e^{-i\gamma d_{uvw} \sigma^z_u \sigma^z_v \sigma^z_w}.\]
Assuming $u,v,w$ are neighboring qubits, $e^{-i\gamma d_{uvw} \sigma^z_u \sigma^z_v \sigma^z_w}$
can be implemented with a Z-rotation gate $R_z(2\gamma d_{uvw}) = e^{-i\gamma d_{uvw} \sigma^z}$
conjugated with stairs of CNOT gates computing and uncomputing the parity of $u,v,w$, 
see Figure~\ref{fig:compilation} (left).
On restricted hardware such as Linear Nearest Neighbor architectures, though, qubits appearing in the same Pauli-Z product are not neighboring qubits. 
In these cases, one would use generalized swap networks~\cite{OGorman2019} to bring
the qubits of each term into neighboring positions at some point during the swap network.
Depending on the number of terms in $H_C$, such a swap network may or may not asymptotically 
increase the size or depth to implement $U_P(\gamma)$, compared to a fully connected architecture.

For the asymptotics of a compilation of multi-controlled phase-shift gate (which is at the heart of a Grover selective phase shift mixer), we largely follow a decomposition 
which gives a linear number of 1- and 2-qubit gates~\cite{Gidney2015-Bootstrapping}, and explain how that method extends to LNN architectures.
In a first step, a multi-controlled $Z^t$-gate can be decomposed into a constant number of Multi-Toffoli, Increment/Decrement and single-qubit $Z^{t/2^i}$ gates, see Figure~\ref{fig:compilation} (right). The Multi-Toffoli gates are interleaved, hence to compile them down to 1- and 2-qubit gates, we can make use of the qubits between the controls as ``borrowed'' ancillas, which need to be returned to their state. 
This can be done for each Multi-Toffoli with a linear number of 1- and 2-qubit gates~\cite{Gidney2015-Toffolis}. 
At the point where the Increment and the Decrement gate come into play, the bottom qubit is ``freed up'' and can be used as a borrowed ancilla as well. With this, we can split an Increment (or a Decrement) gate again into a constant number of Increments and Multi-Toffolis on half the number of qubits~\cite{Gidney2015-Increments}.
As the Increment and the Decrement gate are only used to implement relative phase shifts, one may change the actual order of the wires (as long as the increment follows the original order) -- in this way interleaving the Multi-Toffolis and Increments again with wires of borrowable ancillas. Using these borrowable qubits again enables compilation into a linear number of 1- and 2-qubit gates~\cite{Gidney2015-Increments}, using similar ideas as the Van Rentergem adder~\cite{Rentergem2005}.
We remark that by using additional zeroed ancilla qubits, the constant factors 
of such a construction can significantly be reduced. In summary, we have:
\begin{theorem}
    Given a state preparation unitary $U_S\colon \ket{0}^{\otimes n} \mapsto \ket{F}$,
    the size and depth of a Grover mixing unitary
    $U_M(\beta) = e^{-i\beta\ket{F}\bra{F}} = U_S (\Id-(1-e^{-i\beta})\ket{0}\bra{0}) U_S^{\dagger}$ are bounded by $O(U_S + n)$, even on Linear Nearest Neighbor architectures.
    \label{thm:mixer-size}
\end{theorem}

\section{Applications}
\label{sec:app}
In this section, we illustrate the potential of GM-QAOA on three different problem types:
Optimization problems with non-intersecting constraints such as \textsc{Max K-VertexCover} (with a single equality constraint),
permutation-based optimization problems such as the Traveling Salesperson Problem (\textsc{TSP}),
and applications to discrete portfolio rebalancing. 

To this end, we present for each type a construction of an efficient state preparation unitary $U_S$ that 
prepares an equal superposition of all feasible states, $U_S \ket{0\dots0} = \tfrac{1}{\sqrt{|F|}} \sum_{x\in F} \ket{x}$.
We then compare the resulting Grover Mixer with other Mixers, in terms of 
circuit size and depth as well as structure and quality of the Mixer.

In our constructions, we make use of a recent result~\cite{dickestates} on the preparation of Dicke states $\dicke{n}{k}$, 
which are equal superpositions of all $\binom{n}{k}$ many $n$-qubit computational basis states $x$ with Hamming weight $HW(x)=k$:

\begin{theorem}[Theorem 1 in \cite{dickestates}]
	Dicke states $\dicke{n}{k}$ can be prepared with a circuit of size $\bigO(n\cdot k)$ and depth $\bigO(n)$, 
	even on Linear Nearest Neighbor architectures.
	\label{thm:dicke}
\end{theorem}

In particular, that construction gives a linear-depth and -size circuit to construct $W_n$-states (symmetric $n$-qubit states of Hamming weight 1) 
and can also be used to prepare arbitrary symmetric states with maximum Hamming weight $k$.

\subsubsection*{Application Outline}
In the following subsections, we study \textsc{Max k-VertexCover}, the Traveling Salesperson Problem and Portfolio Rebalancing.
For each of these three applications, we give 
\begin{itemize}
	\item	a short introduction and the problem encoding,
	\item	an explicit construction of the state preparation unitary $U_S$,
		including clickable links to interactive implementations with the drag-and-drop
		quantum circuit simulator Quirk (\url{www.algassert.com/quirk}),
	\item	a discussion of improvements over other approaches,
	\item	and possible generalizations.
\end{itemize}

\subsection{Max k-VertexCover}

We used $k$-Vertex Cover as our running example when introducing GM-QAOA in Section \ref{sec:gm-qaoa} giving a formal definition. The problem is finding a subset $V'$ of exactly $k$ vertices from a given graph, such that the number of edges with at least one end point in $V'$ is maximized.

\subsubsection{Problem encoding}
We let each vertex be represented by a binary variable $x_j$, where $x_j = 1$ if it is in the subset and $x_j = 0$ otherwise for solution $x$.  The objective function is given in graph form in Eq.~\ref{eq:kvccost}. To express this in unitary form, we set
\begin{align}
\begin{split}
	C(x) &= \sum_{(v_j, v_l)\in E} \text{OR}(x_j,x_l) \\
	&= \sum_{(v_j, v_l)\in E} 1 - (1-x_j)(1-x_l)
\end{split}
\label{eq:kvc-or}
\end{align}
Using the transformation $x_j \mapsto (1-\sigma_j^z)/2$, the phase separator  Hamiltonian becomes:
\begin{equation}
	H_C = \frac{1}{4}\sum_{(v_j, v_l)\in E} 3I - \sigma_j^z \sigma_l^z - 
	\sigma_j^z - \sigma_l^z.
	\label{phaseop}
\end{equation}

As $H_C$ simply consists of mutually commuting 1- and 2-local terms which are diagonal in the computational basis, the implementation of the corresponding Phase separator unitary $U_P(\gamma) = e^{-i\gamma H_C}$ is straightforward.

\subsubsection{State preparation}
Since the set  $F$ of feasible solutions contains exactly all solutions $x$ of Hamming weight $k$, we can use the Dicke state construction referenced in the previous section to create unitary $U_S$. As Dicke states can be prepared with $O(nk)$ gates in $O(n)$ depth (see, e.g. 
the preparation of $\dicke{5}{3}$ in an 
\href{https://algassert.com/quirk#circuit=
}{interactive example}), according
to Theorem~\ref{thm:mixer-size} so can a Grover mixer based on Dicke state preparation.

\subsubsection{Improvement}
Previous approaches in QAOA problems with Hamming weight equality constraints
have looked at mixing unitaries based on XY-Hamiltonians $\sigma^x_i \sigma^x_j + \sigma^y_i \sigma^y_j$, namely a ring mixer and a complete mixer, named after 
the pairs of qubits on which the Hamiltonian acts:
\newcommand{\xyring}{H_{XY\text{-Ring}}}
\newcommand{\xyclique}{H_{XY\text{-Clique}}}
\begin{align*}
     \xyring &= \sum_{i} \sigma^x_i \sigma^x_{i+1} + \sigma^y_i \sigma^y_{i+1} \\
     \xyclique &= \sum_{i<j} \sigma^x_i \sigma^x_{j} + \sigma^y_i \sigma^y_{j}
\end{align*}
The corresponding mixing unitaries $e^{-i\beta \xyring}$, $e^{-i\beta \xyclique}$
both preserve the feasible subspace but provide different transitions between feasible
states, due to the ``better connectivity'' of the complete graph mixing Hamiltonian.
An experimental study on instances of \textsc{Max-k-VertexCover}~\cite{eidenbenz2019quantum} has shown that QAOA with a complete mixer outperforms QAOA with a ring mixer, with the advantage most prominent for an intermediate number of QAOA rounds, and for Dicke states as starting states 
(compared to starting in a random but fixed classical Hamming weight $k$ state 
followed by an initial mixer, cf. Figure~\ref{fig:qaoa2}). 
Qualitatively similar results have been found for the \textsc{Max-k-colorable Subgraph} problem, 
in which the color of each vertex is encoded in a one-hot-encoding on which an $XY$-mixer acts~\cite{nasa}.

While $\xyclique$ outperforms $\xyring$, an analytical solution of the $XY$-model on a complete graph is not known, hence an exact implementation of $e^{-i\beta \xyclique}$ seems out of reach.
An implementation which is equivalent in the Hamming weight 1 and Hamming weight $n-1$ subspaces is known, but as such only applicable to \textsc{Max-k-colorable Subgraph} and not to  \textsc{Max-k-VertexCover}.
On the other hand, taking a fermionic view at the $\xyring$ Hamiltonian one gets quadratic fermionic couplings which can be diagonalized using a Fermionic Fast Fourier Transform~\cite{nasa}, which can be implemented in quadratic size and linear depth on Linear Nearest Neighbor architectures using a Givens-rotation network~\cite{givens-FFFT}.

Our Grover mixer $U_M(\beta) = e^{-i\beta H_{GM}} = e^{-i\beta \ket{F}\bra{F}}$ can close this gap, combining better connectivity between feasible states in the mixing Hamiltonian as well as exact implementability in linear depth on LNN architectures. 
For a comparison of $H_{GM}$ with $\xyring$ and $\xyclique$, we can consider the actions of all three mixing unitaries \emph{in the feasible suspace} by looking at the three mixing Hamiltonians restricted to the feasible subspace $F$. These are (up to constant factors)
\begin{align*}
    \xyring|_F  & \cong \sum_{{x,y \in F,\ BSD(x,y)=1}} \ket{x}\bra{y}  \\  
    \xyclique|_F  & \cong\, \sum_{{x,y \in F,\ HD(x,y)=2}}\, \ket{x}\bra{y}  \\  
    H_{GM}  & \cong\hspace*{25pt}\sum_{{x,y \in F}}\hspace*{25pt} \ket{x}\bra{y} 
\end{align*}
where $BSD(x,y)$ is the Bubble Sort distance and $HD(x,y)$ is the Hamming distance  
between strings $x$ and $y$. As any two binary strings of the same Hamming weight 
with Bubble Sort distance 1 must have Hamming distance 2,
$\xyclique|_F$ sums over a superset of $\xyring|_F$ and $H_{GM}$ over a superset of $\xyclique|_F$, with $H_{GM}$ having full connectivity between feasible states while still
having an $O(n)$ depth and $O(n\cdot k) \subseteq O(n^2)$ size circuit implementation for its mixer, combining the advantages of the other two mixing Hamiltonians.

\subsubsection{Generalization}
The state preparation and mixer of GM-QAOA for \textsc{Max-k-VertexCover} are immediately 
applicable to other solution-size equality constraint optimization problems such as  \textsc{Densest-k-Subgraph} or \textsc{Max-k-SetCover}.
However, consideration has to be given to the complexity of the cost Hamiltonians -- 
in the former case, the terms stay 2-local (as the OR of \textsc{Max-k-VertexCover} in Equation~\eqref{eq:kvc-or} is replaced by an AND), 
but for \textsc{Max-k-SetCover} the locality as well as the number of terms in $H_C$ 
might be large (albeit still polynomial in the problem size).

\subsection{Traveling Salesperson Problem}
\label{sec:tsp}

\begin{figure*}[t!]
	\centering
	\begin{adjustbox}{width=\textwidth}
		\def\rowspace{20pt}
		\begin{quantikz}[row sep={\rowspace,between origins},execute at end picture={
				\foreach \col in {9,11,13,15}
				{
					\draw[thick] ($(\tikzcdmatrixname-10-\col.west)-(0pt,0pt)$) to ($(\tikzcdmatrixname-10-\col.east)-(0pt,\rowspace)$);
					\draw[thick] ($(\tikzcdmatrixname-10-\col.west)-(0pt,\rowspace)$) to ($(\tikzcdmatrixname-10-\col.east)-(0pt,2*\rowspace)$);
					\draw[thick] ($(\tikzcdmatrixname-10-\col.west)-(0pt,2*\rowspace)$) to ($(\tikzcdmatrixname-10-\col.east)-(0pt,0pt)$);
				}
			}]
		\lstick{\ket{0}}	& \gate[4]{W_4}	& \ctrl{12}	& \qw		& \qw		& \qw\slice{}	& \qw		& \qw		& \qw		& \qw		& \qw		& \qw		& \qw		& \qw & \qw\slice[style={blue,dotted}]{}& \qw	& \qw		& \qw		& \qw\slice{}	& \qw		& \qw		& \qw		& \qw		& \qw & \qw\slice[style={blue,dotted}]{}& \qw	& \qw		& \qw		& \qw\slice{}		& \qw	\\ 
		\lstick{\ket{0}}& \qw\slice[style={blue,dotted}]{}& \qw	& \ctrl{12}	& \qw		& \qw		& \qw		& \qw		& \qw		& \qw		& \qw		& \qw		& \qw		& \qw		& \qw		& \qw		& \qw		& \qw		& \qw		& \qw		& \qw		& \qw		& \qw		& \qw		& \qw		& \qw		& \qw		& \qw		& \qw		& \qw	\\ 
		\lstick{\ket{0}}	& \qw		& \qw		& \qw		& \ctrl{12}	& \qw		& \qw		& \qw		& \qw		& \qw		& \qw		& \qw		& \qw		& \qw		& \qw		& \qw		& \qw		& \qw		& \qw		& \qw		& \qw		& \qw		& \qw		& \qw		& \qw		& \qw		& \qw		& \qw		& \qw		& \qw	\\ 
		\lstick{\ket{0}}	& \qw		& \qw		& \qw		& \qw		& \ctrl{12}	& \qw		& \qw		& \qw		& \qw		& \qw		& \qw		& \qw		& \qw		& \qw		& \qw		& \qw		& \qw		& \qw		& \qw		& \qw		& \qw		& \qw		& \qw		& \qw		& \qw		& \qw		& \qw		& \qw		& \qw	\\[4pt] 
		\lstick{\ket{0}}	& \qw		& \qw		& \qw		& \qw		& \qw		& \qw		& \swap{7}	& \qw		& \qw		& \qw		& \qw		& \qw		& \qw		& \qw		& \ctrl{8}	& \qw		& \qw		& \qw		& \qw		& \qw		& \qw		& \qw		& \qw		& \qw		& \qw		& \qw		& \qw		& \qw		& \qw	\\ 
		\lstick{\ket{0}}	& \qw		& \qw		& \qw		& \qw		& \qw		& \qw		& \qw		& \qw		& \swap{6}	& \qw		& \qw		& \qw		& \qw		& \qw		& \qw		& \ctrl{8}	& \qw		& \qw		& \qw		& \qw		& \qw		& \qw		& \qw		& \qw		& \qw		& \qw		& \qw		& \qw		& \qw	\\ 
		\lstick{\ket{0}}	& \qw		& \qw		& \qw		& \qw		& \qw		& \qw		& \qw		& \qw		& \qw		& \qw		& \swap{5}	& \qw		& \qw		& \qw		& \qw		& \qw		& \ctrl{8}	& \qw		& \qw		& \qw		& \qw		& \qw		& \qw		& \qw		& \qw		& \qw		& \qw		& \qw		& \qw	\\ 
		\lstick{\ket{0}}	& \qw		& \qw		& \qw		& \qw		& \qw		& \qw		& \qw		& \qw		& \qw		& \qw		& \qw		& \qw		& \swap{4}	& \qw		& \qw		& \qw		& \qw		& \ctrl{8}	& \qw		& \qw		& \qw		& \qw		& \qw		& \qw		& \qw		& \qw		& \qw		& \qw		& \qw	\\[4pt]
		\lstick{\ket{0}}	& \qw		& \qw		& \qw		& \qw		& \qw		& \qw		& \qw		& \qw		& \qw		& \qw		& \qw		& \qw		& \qw		& \qw		& \qw		& \qw		& \qw		& \qw		& \gate[2]{W_2}	& \qw		& \qw		& \gate{W_2}	& \qw		& \gate{W_2}	& \ctrl{4}	& \qw		& \qw		& \qw		& \qw	\\ 
		\lstick{\ket{0}}	& \qw		& \qw		& \qw		& \qw		& \qw		& \gate[3]{W_3}	& \qw		& \gate[3]{}W_3	& \qw		& \gate[3]{}W_3	& \qw		& \gate[3]{}W_3	& \qw		& \gate[3]{}W_3	& \qw		& \qw		& \qw		& \qw		& \qw		& \gate[2]{W_2}	& \qw		& \qw		& \gate{W_2}	& \qw		& \qw		& \ctrl{4}	& \qw		& \qw		& \qw	\\ 
		\lstick{\ket{0}}	& \qw		& \qw		& \qw		& \qw		& \qw		& \qw		& \qw		& \qw		& \qw		& \qw		& \qw		& \qw		& \qw		& \qw		& \qw		& \qw		& \qw		& \qw		& \qw		& \qw		& \gate[2]{W_2}	& \gate{W_2}	& \qw		& \qw		& \qw		& \qw		& \ctrl{4}	& \qw		& \qw	\\ 
		\lstick{\ket{0}}	& \qw		& \qw		& \qw		& \qw		& \qw		& \qw		& \swap{1}	& \qw		& \swap{2}	& \qw		& \swap{3}	& \qw		& \swap{4}	& \qw		& \qw		& \qw		& \qw		& \qw		& \qw		& \qw		& \qw		& \qw		& \gate{W_2}	& \gate{W_2}	& \qw		& \qw		& \qw		& \ctrl{4}	& \qw	\\[4pt]
		\lstick{\ket{0}}	& \targ{}	& \targ{}	& \qw		& \qw		& \qw		& \qw		& \ctrl{}	& \ctrl{-1}	& \qw		& \qw		& \qw		& \qw		& \qw		& \qw		& \targ{}	& \qw		& \qw		& \qw		& \ctrl{}	& \qw		& \qw		& \ctrl{}	& \qw		& \ctrl{}	& \targ{}	& \qw		& \qw		& \qw		& \qw	\\ 
		\lstick{\ket{0}}	& \targ{}	& \qw		& \targ{}	& \qw		& \qw		& \qw		& \qw		& \qw		& \ctrl{}	& \ctrl{-2}	& \qw		& \qw		& \qw		& \qw		& \qw		& \targ{}	& \qw		& \qw		& \ctrl{-4}	& \ctrl{}	& \qw		& \qw		& \ctrl{}	& \qw		& \qw		& \targ{}	& \qw		& \qw		& \qw	\\ 
		\lstick{\ket{0}}	& \targ{}	& \qw		& \qw		& \targ{}	& \qw		& \qw		& \qw		& \qw		& \qw		& \qw		& \ctrl{}	& \ctrl{-3}	& \qw		& \qw		& \qw		& \qw		& \targ{}	& \qw		& \qw		& \ctrl{-4}	& \ctrl{}	& \ctrl{-6}	& \qw		& \qw		& \qw		& \qw		& \targ{}	& \qw		& \qw	\\ 
		\lstick{\ket{0}}	& \targ{}	& \qw		& \qw		& \qw		& \targ{}	& \qw		& \qw		& \qw		& \qw		& \qw		& \qw		& \qw		& \ctrl{}	& \ctrl{-4}	& \qw		& \qw		& \qw		& \targ{}	& \qw		& \qw		& \ctrl{-4}	& \qw		& \ctrl{-6}	& \ctrl{-7}	& \qw		& \qw		& \qw		& \targ{}	& \qw	 
	\end{quantikz}
	\end{adjustbox}
	\caption{State Preparation $U_S$: 
		(left) Initialization of the first row in a $W_4$-state and $\ket{1111}$ in the bitmask in the last row, followed by a bitmask update.
		(middle) Swapping a $W_3$-state from the third row into the second row, controlled on the bitmask, followed by a bitmask update.
		(right) Direct $W_2$-state generation in the third row, controlled on the bitmask, followed by a bitmask update. 
		At the end, the row constraint of the fourth row is also satisfied. 
	}
	\label{fig:permutation-superposition}
\end{figure*}
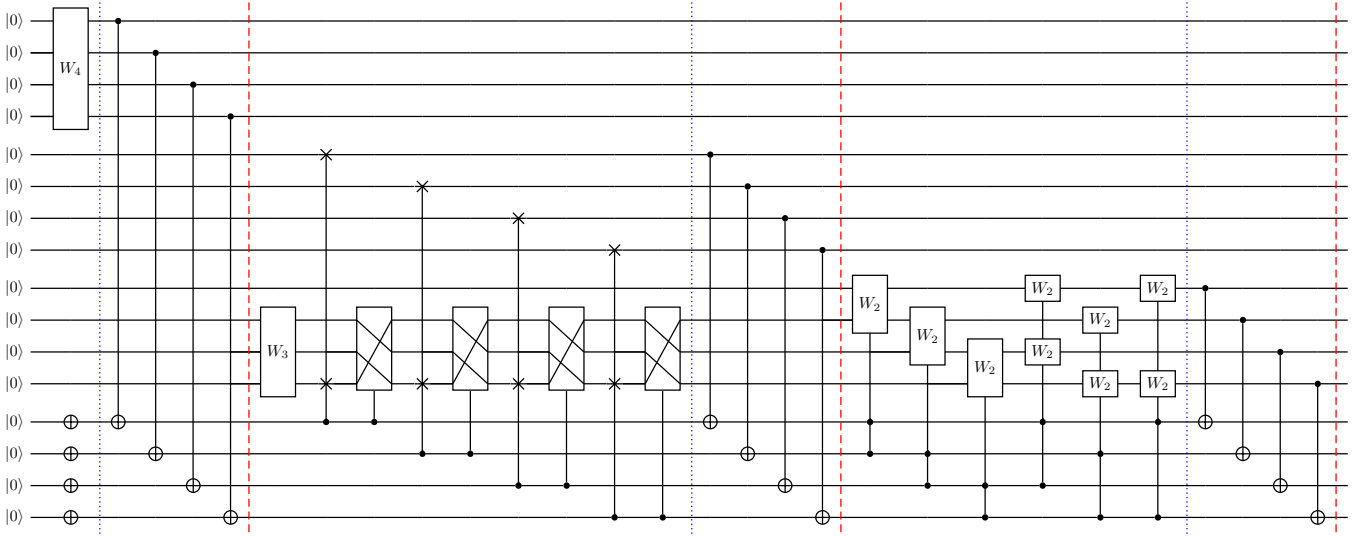

In permutation-based $NP$ optimization problems the task is to find a minimum/maximum objective value over all permutations of the input.
For example, in the Traveling Salesperson Problem, 
we are given $n$ cities with a distance function $d\colon [n]\times[n] \rightarrow \mathbb{R}_{\geq 0}$ between each pair of cities.
The goal is to find a tour through all cities that minimizes (approximately or exactly) the total travel distance along the tour.

\subsubsection{Problem encoding}
Such a tour is given as a permutation $\tau\colon [n] \rightarrow [n]$ of the input cities, 
and we get a minimization problem over all permutations in the symmetric group $S_n$:
\[ \min_{\tau \in S_n}{ \sum_{k=1}^{n} d(\tau^{-1}(k), \tau^{-1}(k+1)) }.	\]

We can represent a permutation $\tau$ by a permutation matrix, 
i.e. a matrix which has exactly one ``1'' entry per column and per row (and ``0''-entries otherwise). 
Here, a ``1'' entry in row $r$ and column $c$ denotes city $c$ appearing at the $r$-th position of the tour. 
Equivalently, we may consider $\tau$ as a bitstring $x = x_{n^2-1}\dots x_1x_0$ (with 0-based indexing), 
where the $r$-th group of $n$ bits denotes a one-hot-encoding of $\tau^{-1}(r)$, 
meaning that city $\tau^{-1}(r)$ is mapped to the $r$-th position in the tour. 
Hence we can express the symmetric group as:
\begin{align*}
	S_n = \{ x \mid 
		&\ x = x_{n^2-1}\ldots x_1x_0 \text{ such that}	\\
		&\ \forall\ 0 \leq c < n\colon\ \sum_{\mathclap{j \equiv c \bmod n}} \ x_j = 1	& \text{(col. constraints)}\;\; \\ 
		&\ \forall\ 0 \leq r < n\colon\ \sum_{\mathclap{\lfloor i/n \rfloor = r}} \ x_i = 1	& \text{(row constraints)} \} 
\end{align*}
Using this definition and computational basis states $x\in S_n$ on $n^2$ qubits we get a cost Hamiltonian
\begin{align}
	H_C = \sum_{i=0}^{n-1} \sum_{u=0}^{n-1} \sum_{v=0}^{n-1} d(u,v) \sigma_{i\cdot n + u}^z  \sigma_{(i+1)\cdot n + v}^z,	
	\label{eq:hc-tsp}
\end{align}
adding a distance of $d(u,v)$ once if and only if city $v$ is visited directly after city $u$. 
As $H_C$ simply consists of mutually commuting 2-local terms which are diagonal in the computational basis, 
the implementation of the corresponding Phase separator unitary $U_P(\gamma) = e^{-i\gamma H_C}$ is straightforward.

\subsubsection{State preparation}
We now give an efficient circuit for a state preparation unitary $U_S$ creating an equal superposition of
all $n!$ permutations, $U_S \ket{0}^{\otimes n^2} = \ket{S_n} = \tfrac{1}{\sqrt{n!}} \sum_{x \in S_n} \ket{x}$.
The existence of a polynomial-size circuit to prepare $\ket{S_n}$ (at least approximately) follows from results on efficiently specifiable polynomials including the permanent $\mathrm{perm}(x_0,\ldots,x_{n^2-1}) = \sum_{\tau \in S_n} \prod_{i=1}^n x_{i,\tau(i)}$~\cite[Theorem 28]{fefferman2016}, or by considering the rows of $S_n$ as qudits and using the (inverse) Schur transform for symmetric qudit states~\cite{Bacon2006}.
Here, we give an explicit construction of an $O(n^2)$ depth circuit with $O(n^3)$ gates to prepare $\ket{S_n}$ exactly.

We construct the superposition inductively row by row, such that after processing $k$ rows
our states satisfies \emph{all} column constraints as well as the \emph{first $k$} row constraints.
In order for the column constraints to match up, we use the $n$ qubits of the last row as a \emph{bitmask}
of the yet unoccupied columns, i.e. we get an equal superposition of all $\tfrac{n!}{k!}$ terms in the set
\begin{align*}
	\{ x \mid 
	&\ x = x_{n^2-1}\ldots x_1x_0 \text{ such that}	\\
	&\ \forall\ 0 \leq c < n\colon\ \sum_{\mathclap{j \equiv c \bmod n}} \ x_j = 1	& \text{(col. constraints)}\;\; \\ 
	&\ \forall\ 0 \leq r < \textcolor{red}{k}\colon\ \sum_{\mathclap{\lfloor i/n \rfloor = r}} \ x_i = 1 & \text{(row constraints)} \;\; \\
	&\ \phantom{\forall\ } r = n-1\colon\ \sum_{\mathclap{\lfloor i/n \rfloor = r}} \ x_i = n-k & \text{(bitmask)} \} 
\end{align*}
(Note that for $k=n-1$ we already get the set $S_n$.)
Inductively, the $k$-th row is processed in two steps, illustrated for all rows in Figure~\ref{fig:permutation-superposition}:
\begin{itemize}
	\item	First, the row constraint for the new row is satisfied.
	\item	Next, the bitmask in the last row is updated.
\end{itemize}

In more detail, before adding the row of index $k-1$, the bitmask row designates (in superposition) 
the $n-k+1$ columns of the permutation matrix in which a ``1'' may still be entered.
For the very first row (Figure~\ref{fig:permutation-superposition} left), the bitmask will simply be $\ket{1}^{\otimes n}$, and we easily create a $W_n$-state on the first $n$ qubits.
For the second last row of index $n-2$ (Figure~\ref{fig:permutation-superposition} right), the bitmask is a superposition of all $n$-qubit states of Hamming weight 2, hence 
for each of the $\binom{n}{2}$ possibilities we use a $2$-controlled $W_2$-state preparation on the respective columns.  

For any row of index $k-1$ in-between, such a directly controlled $W_{n-k+1}$-state preparations would result in $\binom{n}{n-k+1}$ many controlled operations.
Instead, by using the following row $k$ as an ancilla row, we can significantly reduce the number of operations.  
To this end, we generate a $W_{n-k+1}$-state on the last $n-k+1$ qubits of row $k$, and swap these qubits into the correct places in row $k-1$
(Figure~\ref{fig:permutation-superposition} middle). This is done with 
\begin{itemize}
	\item	$n$ controlled swaps (for $0\leq i < n$ swapping qubits $k n-1$ and $(k-1)n+i$, controlled by qubit $(n-1)n+i$),
	\item	each followed by a controlled cyclic rotation of the last $n-k+1$ qubits of row $k$ (controlled by the same qubit $(n-1)n+i$). 
\end{itemize}
The controls come from the bitmask register, hence exactly $n-k+1$ of these operations will be executed, 
at which point row $k$ is again in the state $\ket{0}^{\otimes n}$.
For each row, after adding the new row constraint, we update the bitmask with $n$ CNOT gates, 
for $0\leq i < n$ resetting qubit $(n-1)n + i$ controlled on qubit $(k-1)n + i$.

Circuit implementations in Quirk can be found 
\href{https://algassert.com/quirk#circuit=%7B%22cols%22%3A%5B%5B%22~94a8%22%2C1%2C1%2C1%2C1%2C1%2C%22X%22%2C%22X%22%2C%22X%22%5D%2C%5B%22Chance3%22%2C1%2C1%2C%22Chance3%22%2C1%2C1%2C%22Chance3%22%5D%2C%5B%22%E2%80%A2%22%2C1%2C1%2C1%2C1%2C1%2C%22X%22%5D%2C%5B1%2C%22%E2%80%A2%22%2C1%2C1%2C1%2C1%2C1%2C%22X%22%5D%2C%5B1%2C1%2C%22%E2%80%A2%22%2C1%2C1%2C1%2C1%2C1%2C%22X%22%5D%2C%5B%22Chance3%22%2C1%2C1%2C%22Chance3%22%2C1%2C1%2C%22Chance3%22%5D%2C%5B1%2C1%2C1%2C%22~mv19%22%2C1%2C1%2C%22%E2%80%A2%22%2C%22%E2%80%A2%22%5D%2C%5B1%2C1%2C1%2C1%2C%22~mv19%22%2C1%2C1%2C%22%E2%80%A2%22%2C%22%E2%80%A2%22%5D%2C%5B1%2C1%2C1%2C%22~bvbm%22%2C1%2C1%2C%22%E2%80%A2%22%2C1%2C%22%E2%80%A2%22%5D%2C%5B%22Chance3%22%2C1%2C1%2C%22Chance3%22%2C1%2C1%2C%22Chance3%22%5D%2C%5B1%2C1%2C1%2C%22%E2%80%A2%22%2C1%2C1%2C%22X%22%5D%2C%5B1%2C1%2C1%2C1%2C%22%E2%80%A2%22%2C1%2C1%2C%22X%22%5D%2C%5B1%2C1%2C1%2C1%2C1%2C%22%E2%80%A2%22%2C1%2C1%2C%22X%22%5D%2C%5B%22Chance3%22%2C1%2C1%2C%22Chance3%22%2C1%2C1%2C%22Chance3%22%5D%5D%2C%22gates%22%3A%5B%7B%22id%22%3A%22~590m%22%2C%22name%22%3A%22%E2%88%9A1%2F3%22%2C%22matrix%22%3A%22%7B%7B%E2%88%9A%E2%85%93%2C-%E2%88%9A%E2%85%94%7D%2C%7B%E2%88%9A%E2%85%94%2C%E2%88%9A%E2%85%93%7D%7D%22%7D%2C%7B%22id%22%3A%22~94a8%22%2C%22name%22%3A%22W3%22%2C%22circuit%22%3A%7B%22cols%22%3A%5B%5B%22X%22%5D%2C%5B%22%E2%80%A2%22%2C%22~590m%22%5D%2C%5B%22X%22%2C%22%E2%80%A2%22%2C%22H%22%5D%2C%5B1%2C%22X%22%2C%22%E2%80%A2%22%5D%5D%7D%7D%2C%7B%22id%22%3A%22~mv19%22%2C%22name%22%3A%22W2%22%2C%22circuit%22%3A%7B%22cols%22%3A%5B%5B%22X%22%5D%2C%5B%22%E2%80%A2%22%2C%22H%22%5D%2C%5B%22X%22%2C%22%E2%80%A2%22%5D%5D%7D%7D%2C%7B%22id%22%3A%22~bvbm%22%2C%22name%22%3A%22W2'%22%2C%22circuit%22%3A%7B%22cols%22%3A%5B%5B%22X%22%5D%2C%5B%22%E2%80%A2%22%2C1%2C%22H%22%5D%2C%5B%22X%22%2C1%2C%22%E2%80%A2%22%5D%5D%7D%7D%5D%7D
}{here (for $S_3$)}, and
\href{https://algassert.com/quirk#circuit=
}{here (for $S_4$)}.
Finally, we note that there is a redundancy in the superposition of cyclically shifted permutations -- they encode the same tour.
One may circumvent this by choosing a fixed position in the first row, instead of generating a $W_n$ state. We conclude:

\begin{theorem}
	A state preparation unitary $U_S$ for an equal superposition of all permutations in $S_n$, 
	$U_S\colon \ket{0}^{\otimes n^2} \mapsto |S_n|^{-1/2} \sum_{x \in S_n} \ket{x}$
	can be implemented with $O(n^3)$ gates in depth $O(n^2)$, even on Linear Nearest Neighbor architectures.  
	\label{thm:tsp-state-preparation}
\end{theorem}

\begin{proof}
	We prove that the processing of each row can be implemented with $O(n^2)$ gates in depth $O(n)$:
	For the second to last row, we have $O(n^2)$ 4-qubit gates, each of which can be implemented in constant size 
	in parallel with $\approx \tfrac{n}{2}$ other of these gates for a row depth of $O(n)$.
	
	For all other rows, the process consist of a $W$-state generation and $n$ controlled swaps and cyclic rotations.
	$W$-states can be implemented in $O(n)$ size and depth, and cyclic rotations as well
	(via a stair of swaps between neighboring qubits). Hence these larger stair-shaped gates can be pushed together 
	to overlap in depth, keeping a linear $O(n)$ depth overall, even on LNN architectures.

	Finally, to have the same asymptotics overall for LNN architectures, one can keep the bitmask register not in 
	the bottom/last row, but always adjacent to the currently processed row. Shifting the bitmask register to lower 
	rows over time needs an additional $O(n^3)$ gates in $O(n^2)$ depth in total, keeping the same asymptotics.
\end{proof}

\subsubsection{Improvement}
In the following we show that our Grover Mixer approach is competive in depth with existing mixer designs for permutations,
while lowering the number of total gates and not suffering from side effects of Trotterization. 

Previous approaches to design mixers for permutations have been based on Hamiltonians $H_{M,\{i,j\},\{u,v\}}$ swapping 
cities $u,v$ in tour positions $i,j$~\cite{Hadfield2019}. These can be combined into value-independent ordering swap 
mixers, most prominently for transpositions (adjacent rows): 
\[	U_{M,i}(\beta) = e^{-i\beta H_{M,i}} \text{ with } H_{M,i} = \sum\nolimits_{u,v} H_{M,\{i,i+1\},\{u,v\}},	\]
Such a mixer can be implemented with $\binom{n}{2}$ four-qubit gates, one for each of the terms $H_{M,\{i,i+1\},\{u,v\}}$,
where the four-qubit gates can be grouped in layers of $\approx n/2$ gates which can be implemented in parallel for 
a depth of $O(n)$. Using a generalized swap network~\cite{OGorman2019},  
$O(n^2)$ gates in $O(n)$ depth can even be achieved on LNN architectures.

However, two such Hamiltonians $H_{M,i}, H_{M,i+1}$ do not commute, hence exactly implementing a  
simultaneous ordering swap mixer $U_M(\beta) = e^{-i\beta \sum_i H_{M,i}}$ seems out of reach. 
Thus one may resort to a Trotterization implementing each $U_{M,i}(\beta)$ individually. 
To get transitions between all pairs of permutations, note that any permuation $\sigma$ can be reached from
any other permutation $\sigma'$ by at most $\binom{n}{2}$ transpositions, of which up to $\approx n/2$ 
can be grouped to be executed in parallel. For example, a linear number of layers of alternatingly 
all $U_{M,i}$ with odd $i$ and all $U_{M,i}$ with even $i$ will give a mixer with $O(n^4)$ gates and $O(n^2)$ depth
that preserves the feasible subspace of all permutations, 
and provides transitions between all pairs of permutations $\sigma, \sigma'$. 
Other partitions of four-qubit gates into layers may also be possible.
In any case, in comparison to a simultaenous swap mixer, the transitions do not only depend on $\beta$,
but also suffer an unwanted Trotterization error that depends on the exact partitioning of four-qubit gates. 

We can compare this to our Grover Mixer $U_M(\beta) = e^{-i\beta \ket{S_n}\bra{S_n}}$, which can be implemented
exactly with $O(n^3)$ gates in $O(n^2)$ depth on LNN architectures and provides transitions between all pairs 
of permutations without any implementation artifacts.
This also matches the numbers to implement the phase separator $U_P(\gamma) = e^{-i\gamma H_C}$ which can 
be implemented with $O(n^3)$ gates in $O(n^2)$ depth (even on LNN architectures, when using swap networks).
Hence a Grover Mixer improves on previously considered ordering swap mixers in terms of the total number of gates, while matching their depth.

\subsubsection{Generalizations}
Our GM-QAOA approach to the Traveling Salesperson Problem extends to a variety of other problems based on injective mappings. 
We give two examples:
\begin{itemize}
	\item	The Quadratic Assignment Problem~\cite{QAP} asks for a bijective assignment $f\colon L \rightarrow P$ from 
		a set of locations $L$ (with a distance function $d\colon L\times L \rightarrow \mathbb{R}_{\geq 0}$) to 
		a set of facilities $P$ (with a weight function $w\colon P\times P \rightarrow \mathbb{R}_{\geq 0}$)
		to minimize weighted distances between facilities,
		\[	\min_f \sum_{u,v \in P} w(u,v) \cdot d(f^{-1}(u), f^{-1}(p)).	\]
		This can be seen as a generalization of TSP, where we have $w(u,u+1) = 1$ and $w(u,v)=0$ otherwise.
	\item	The Maximum Common Edge Subgraph Promlem~\cite{MCES} asks, given graphs $G=(V_G,E_G)$ and $H=(V_H,E_H)$ 
		with $|V_G|\leq |V_H|$, for a common subgraph of $G$ and $H$ with a maximum number of edges. 
		This can be seen as a generalization of the graph isomorphism problem, 
		i.e. we can formulate it as looking for an injective function
		$f\colon V_G \rightarrow V_H$ matching as many edges as possible. 
\end{itemize}

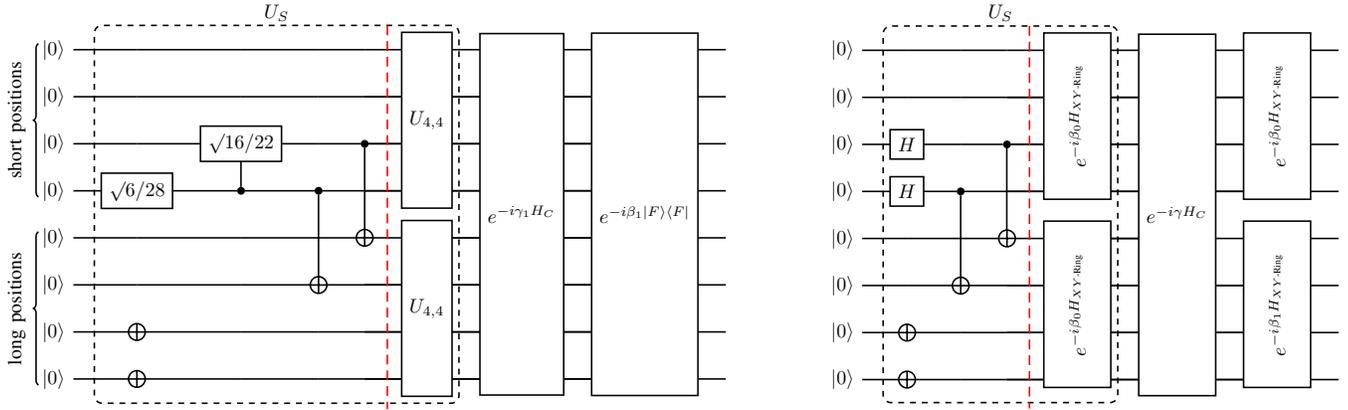
\begin{figure*}[t!]
	\centering
	\begin{adjustbox}{width=\linewidth}
	\newcommand{\sgate}[1]{\gate{\surd#1}}
	\newcommand{\ugate}[1]{\gate[#1]{U_{#1,#1}}}
	\begin{quantikz}[row sep={24pt,between origins}]
		\lstick[4]{\rotatebox{90}{short positions}}
		&\lstick{\ket{0}}	
		& \qw\gategroup[8,steps=5,style={dashed,rounded corners, inner sep=0pt},background]{$U_S$}
		& \qw
		& \qw
		& \qw
		& \ugate{4}
		& \gate[8]{e^{-i\gamma_1 H_C}}
		& \gate[8]{e^{-i\beta_1 \ket{F}\bra{F}}}
		& \qw
		\\	
		&\lstick{\ket{0}}	& \qw			&\qw		&\qw		&\qw\slice{}	&\qw		&\qw		&\qw		&\qw		\\	
		&\lstick{\ket{0}}	& \qw			&\sgate{16/22}	&\qw		&\ctrl{2}	&\qw		&\qw		&\qw		&\qw		\\	
		&\lstick{\ket{0}}	& \sgate{6/28}		&\ctrl{-1}	&\ctrl{2}	&\qw		&\qw		&\qw		&\qw		&\qw		\\	
		\lstick[4]{\rotatebox{90}{long positions}}
		&\lstick{\ket{0}}	& \qw			&\qw		&\qw		&\targ{}	&\ugate{4}	&\qw		&\qw		&\qw		\\	
		&\lstick{\ket{0}}	& \qw			&\qw		&\targ{}	&\qw		&\qw		&\qw		&\qw		&\qw		\\	
		&\lstick{\ket{0}}	& \targ{}		&\qw		&\qw		&\qw		&\qw		&\qw		&\qw		&\qw		\\	
		&\lstick{\ket{0}}	& \targ{}		&\qw		&\qw		&\qw		&\qw		&\qw		&\qw		&\qw		
	\end{quantikz}
	\qquad\qquad
	\newcommand{\mgate}[1]{\gate[4]{\rotatebox{90}{$\mathclap{e^{-i\beta_{#1} \xyring}}$}}U_{4,4}}
	\begin{quantikz}[row sep={24pt,between origins}]
		\lstick{\ket{0}}	
		& \qw\gategroup[8,steps=4,style={dashed,rounded corners, inner sep=0pt},background]{$U_S$}
		& \qw
		& \qw\slice{}
		& \mgate{0}
		& \gate[8]{e^{-i\gamma H_C}}
		& \mgate{0}
		& \qw
		\\	
		\lstick{\ket{0}}	& \qw			&\qw		&\qw		&\qw		&\qw		&\qw		&\qw		\\	
		\lstick{\ket{0}}	& \gate{H}		&\qw		&\ctrl{2}	&\qw		&\qw		&\qw		&\qw		\\	
		\lstick{\ket{0}}	& \gate{H}		&\ctrl{2}	&\qw		&\qw		&\qw		&\qw		&\qw		\\	
		\lstick{\ket{0}}	& \qw			&\qw		&\targ{}	&\mgate{0}	&\qw		&\mgate{1}	&\qw		\\	
		\lstick{\ket{0}}	& \qw			&\targ{}	&\qw		&\qw		&\qw		&\qw		&\qw		\\	
		\lstick{\ket{0}}	& \targ{}		&\qw		&\qw		&\qw		&\qw		&\qw		&\qw		\\	
		\lstick{\ket{0}}	& \targ{}		&\qw		&\qw		&\qw		&\qw		&\qw		&\qw		
	\end{quantikz}
	\end{adjustbox}
	\caption{Discrete Portfolio Rebalancing QAOA where we invest in long and short positions for a net total of $2$ discrete lots, 
	with bands $2-0 = 3-1 = 4-2 =2$.
	(left) GM-QAOA: We first create a superposition 
	assigning the correct weights to each band,
	$\sqrt{\tfrac{6}{28}}\ket{1100}\ket{0000}+ \sqrt{\tfrac{16}{28}}\ket{1110}\ket{1000}+ \sqrt{\tfrac{6}{28}}\ket{1111}\ket{1100}$.
	Fed into two Dicke state preparation unitaries $U_{4,4}$, this gives the equal superposition over all combinations of long and short positions for 2 discrete lots.
	(right) Previous approach: Bell states of short and long positions give the correct net total, but lead to a binomial weight distribution over the different bands.
	Furthermore, using Hamming-weight preserving $XY$-ring mixers for both short and long positions individually prohibits a mixing between different bands.}
	\label{fig:portfolio}
\end{figure*}

Finally, we remark that our approach can also cover problems based on the alternating group $A_n$, 
i.e. the set of all permutations $\sigma$ with an even number of inversions (elements $i,j$ with $\sigma(i)>\sigma(j)$). 
For an permutation in string notation $x = x_{n^2-1}\dots x_1x_0 \in A_n$, 
this means that the number of pairs $r<c$, for which $x_{r\cdot n+c} = 1$ holds, must be even. 

This number, modulo 2, can easily be computed in an ancilla in the following way: 
After the completion of row $r$ and the following bitmask update, 
the number of inversions due to row $r$ can be computed from row $r$ and the bitmask 
(in linear depth, if we transform the bitmask into a prefix bitmask first, and uncompute it again).  
Then, when processing the second last row, we no longer generate controlled $W_2$-states, but actually 
only a $\ket{1}$ in the correct place, such that the last two rows complete an even permutation.
This can be done controlled on the bitmask \emph{and} the ancilla qubit, followed by an 
uncomputation of the ancilla before the last bitmask update. 
For an interactive example, see
\href{https://algassert.com/quirk#circuit=%7B%22cols%22%3A%5B%5B%22~94a8%22%2C1%2C1%2C1%2C1%2C1%2C%22X%22%2C%22X%22%2C%22X%22%5D%2C%5B%22Chance3%22%2C1%2C1%2C%22Chance3%22%2C1%2C1%2C%22Chance3%22%2C1%2C1%2C%22Chance%22%5D%2C%5B%22%E2%80%A2%22%2C1%2C1%2C1%2C1%2C1%2C%22X%22%5D%2C%5B1%2C%22%E2%80%A2%22%2C1%2C1%2C1%2C1%2C1%2C%22X%22%5D%2C%5B1%2C1%2C%22%E2%80%A2%22%2C1%2C1%2C1%2C1%2C1%2C%22X%22%5D%2C%5B%22Chance3%22%2C1%2C1%2C%22Chance3%22%2C1%2C1%2C%22Chance3%22%2C1%2C1%2C%22Chance%22%5D%2C%5B1%2C1%2C1%2C1%2C1%2C1%2C%22%E2%80%A2%22%2C%22X%22%5D%2C%5B1%2C1%2C1%2C1%2C1%2C1%2C1%2C%22%E2%80%A2%22%2C%22X%22%5D%2C%5B1%2C%22%E2%80%A2%22%2C1%2C1%2C1%2C1%2C%22%E2%80%A2%22%2C1%2C1%2C%22X%22%5D%2C%5B1%2C1%2C%22%E2%80%A2%22%2C1%2C1%2C1%2C1%2C%22%E2%80%A2%22%2C1%2C%22X%22%5D%2C%5B1%2C1%2C1%2C1%2C1%2C1%2C1%2C%22%E2%80%A2%22%2C%22X%22%5D%2C%5B1%2C1%2C1%2C1%2C1%2C1%2C%22%E2%80%A2%22%2C%22X%22%5D%2C%5B%22Chance3%22%2C1%2C1%2C%22Chance3%22%2C1%2C1%2C%22Chance3%22%2C1%2C1%2C%22Chance%22%5D%2C%5B1%2C1%2C1%2C1%2C%22X%22%2C1%2C%22%E2%80%A2%22%2C%22%E2%80%A2%22%2C1%2C%22%E2%80%A2%22%5D%2C%5B1%2C1%2C1%2C%22X%22%2C1%2C1%2C%22%E2%80%A2%22%2C%22%E2%80%A2%22%2C1%2C%22%E2%97%A6%22%5D%2C%5B1%2C1%2C1%2C1%2C1%2C%22X%22%2C1%2C%22%E2%80%A2%22%2C%22%E2%80%A2%22%2C%22%E2%80%A2%22%5D%2C%5B1%2C1%2C1%2C1%2C%22X%22%2C1%2C1%2C%22%E2%80%A2%22%2C%22%E2%80%A2%22%2C%22%E2%97%A6%22%5D%2C%5B1%2C1%2C1%2C1%2C1%2C%22X%22%2C%22%E2%80%A2%22%2C1%2C%22%E2%80%A2%22%2C%22%E2%80%A2%22%5D%2C%5B1%2C1%2C1%2C%22X%22%2C1%2C1%2C%22%E2%80%A2%22%2C1%2C%22%E2%80%A2%22%2C%22%E2%97%A6%22%5D%2C%5B%22Chance3%22%2C1%2C1%2C%22Chance3%22%2C1%2C1%2C%22Chance3%22%2C1%2C1%2C%22Chance%22%5D%2C%5B1%2C1%2C1%2C1%2C%22%E2%80%A2%22%2C1%2C%22%E2%80%A2%22%2C%22%E2%80%A2%22%2C1%2C%22X%22%5D%2C%5B1%2C1%2C1%2C1%2C1%2C%22%E2%80%A2%22%2C1%2C%22%E2%80%A2%22%2C%22%E2%80%A2%22%2C%22X%22%5D%2C%5B1%2C1%2C1%2C1%2C1%2C%22%E2%80%A2%22%2C%22%E2%80%A2%22%2C1%2C%22%E2%80%A2%22%2C%22X%22%5D%2C%5B%22Chance3%22%2C1%2C1%2C%22Chance3%22%2C1%2C1%2C%22Chance3%22%2C1%2C1%2C%22Chance%22%5D%2C%5B1%2C1%2C1%2C%22%E2%80%A2%22%2C1%2C1%2C%22X%22%5D%2C%5B1%2C1%2C1%2C1%2C%22%E2%80%A2%22%2C1%2C1%2C%22X%22%5D%2C%5B1%2C1%2C1%2C1%2C1%2C%22%E2%80%A2%22%2C1%2C1%2C%22X%22%5D%2C%5B%22Chance3%22%2C1%2C1%2C%22Chance3%22%2C1%2C1%2C%22Chance3%22%2C1%2C1%2C%22Chance%22%5D%2C%5B%22%E2%80%A6%22%5D%2C%5B%22Chance%22%2C%22Chance%22%2C%22Chance%22%2C%22Chance%22%2C%22Chance%22%2C%22Chance%22%2C%22Chance%22%2C%22Chance%22%2C%22Chance%22%5D%2C%5B%22Bloch%22%2C%22Bloch%22%2C%22Bloch%22%2C%22Bloch%22%2C%22Bloch%22%2C%22Bloch%22%2C%22Bloch%22%2C%22Bloch%22%2C%22Bloch%22%5D%2C%5B%22Amps9%22%5D%5D%2C%22gates%22%3A%5B%7B%22id%22%3A%22~590m%22%2C%22name%22%3A%22%E2%88%9A1%2F3%22%2C%22matrix%22%3A%22%7B%7B%E2%88%9A%E2%85%93%2C-%E2%88%9A%E2%85%94%7D%2C%7B%E2%88%9A%E2%85%94%2C%E2%88%9A%E2%85%93%7D%7D%22%7D%2C%7B%22id%22%3A%22~94a8%22%2C%22name%22%3A%22W3%22%2C%22circuit%22%3A%7B%22cols%22%3A%5B%5B%22X%22%5D%2C%5B%22%E2%80%A2%22%2C%22~590m%22%5D%2C%5B%22X%22%2C%22%E2%80%A2%22%2C%22H%22%5D%2C%5B1%2C%22X%22%2C%22%E2%80%A2%22%5D%5D%7D%7D%2C%7B%22id%22%3A%22~mv19%22%2C%22name%22%3A%22W2%22%2C%22circuit%22%3A%7B%22cols%22%3A%5B%5B%22X%22%5D%2C%5B%22%E2%80%A2%22%2C%22H%22%5D%2C%5B%22X%22%2C%22%E2%80%A2%22%5D%5D%7D%7D%2C%7B%22id%22%3A%22~bvbm%22%2C%22name%22%3A%22W2'%22%2C%22circuit%22%3A%7B%22cols%22%3A%5B%5B%22X%22%5D%2C%5B%22%E2%80%A2%22%2C1%2C%22H%22%5D%2C%5B%22X%22%2C1%2C%22%E2%80%A2%22%5D%5D%7D%7D%5D%7D
}{here (for $A_3$)}.
With a bit more work, the ancilla can even be included in the first qubit of the second last row, 
for an interactive example, see
\href{https://algassert.com/quirk#circuit=%7B%22cols%22%3A%5B%5B%22~94a8%22%2C1%2C1%2C1%2C1%2C1%2C%22X%22%2C%22X%22%2C%22X%22%5D%2C%5B%22Chance3%22%2C1%2C1%2C%22Chance3%22%2C1%2C1%2C%22Chance3%22%2C1%2C1%2C%22Chance%22%5D%2C%5B%22%E2%80%A2%22%2C1%2C1%2C1%2C1%2C1%2C%22X%22%5D%2C%5B1%2C%22%E2%80%A2%22%2C1%2C1%2C1%2C1%2C1%2C%22X%22%5D%2C%5B1%2C1%2C%22%E2%80%A2%22%2C1%2C1%2C1%2C1%2C1%2C%22X%22%5D%2C%5B%22Chance3%22%2C1%2C1%2C%22Chance3%22%2C1%2C1%2C%22Chance3%22%2C1%2C1%2C%22Chance%22%5D%2C%5B1%2C1%2C1%2C1%2C1%2C1%2C%22%E2%80%A2%22%2C%22X%22%5D%2C%5B1%2C1%2C1%2C1%2C1%2C1%2C1%2C%22%E2%80%A2%22%2C%22X%22%5D%2C%5B1%2C%22%E2%80%A2%22%2C1%2C%22X%22%2C1%2C1%2C%22%E2%80%A2%22%5D%2C%5B1%2C1%2C%22%E2%80%A2%22%2C%22X%22%2C1%2C1%2C1%2C%22%E2%80%A2%22%5D%2C%5B1%2C1%2C1%2C1%2C1%2C1%2C1%2C%22%E2%80%A2%22%2C%22X%22%5D%2C%5B1%2C1%2C1%2C1%2C1%2C1%2C%22%E2%80%A2%22%2C%22X%22%5D%2C%5B%22Chance3%22%2C1%2C1%2C%22Chance3%22%2C1%2C1%2C%22Chance3%22%2C1%2C1%2C%22Chance%22%5D%2C%5B1%2C1%2C1%2C%22%E2%80%A2%22%2C%22X%22%2C1%2C%22%E2%80%A2%22%2C%22%E2%80%A2%22%5D%2C%5B1%2C1%2C1%2C%22%E2%80%A2%22%2C1%2C%22X%22%2C%22%E2%80%A2%22%2C1%2C%22%E2%80%A2%22%5D%2C%5B1%2C1%2C1%2C%22%E2%80%A2%22%2C1%2C%22X%22%2C1%2C%22%E2%80%A2%22%2C%22%E2%80%A2%22%5D%2C%5B1%2C1%2C1%2C%22%E2%97%A6%22%2C%22X%22%2C1%2C1%2C%22%E2%80%A2%22%2C%22%E2%80%A2%22%5D%2C%5B%22Chance3%22%2C1%2C1%2C%22Chance3%22%2C1%2C1%2C%22Chance3%22%2C1%2C1%2C%22Chance%22%5D%2C%5B1%2C1%2C1%2C%22X%22%2C1%2C1%2C%22%E2%80%A2%22%2C%22%E2%80%A2%22%5D%2C%5B1%2C1%2C1%2C%22X%22%2C1%2C1%2C%22%E2%80%A2%22%2C1%2C%22%E2%80%A2%22%5D%2C%5B1%2C1%2C1%2C%22X%22%2C1%2C%22%E2%80%A2%22%2C1%2C%22%E2%80%A2%22%2C%22%E2%80%A2%22%5D%2C%5B%22Chance3%22%2C1%2C1%2C%22Chance3%22%2C1%2C1%2C%22Chance3%22%2C1%2C1%2C%22Chance%22%5D%2C%5B1%2C1%2C1%2C%22%E2%80%A2%22%2C1%2C1%2C%22X%22%5D%2C%5B1%2C1%2C1%2C1%2C%22%E2%80%A2%22%2C1%2C1%2C%22X%22%5D%2C%5B1%2C1%2C1%2C1%2C1%2C%22%E2%80%A2%22%2C1%2C1%2C%22X%22%5D%2C%5B%22Chance3%22%2C1%2C1%2C%22Chance3%22%2C1%2C1%2C%22Chance3%22%2C1%2C1%2C%22Chance%22%5D%2C%5B%22%E2%80%A6%22%5D%2C%5B%22Chance%22%2C%22Chance%22%2C%22Chance%22%2C%22Chance%22%2C%22Chance%22%2C%22Chance%22%2C%22Chance%22%2C%22Chance%22%2C%22Chance%22%5D%2C%5B%22Bloch%22%2C%22Bloch%22%2C%22Bloch%22%2C%22Bloch%22%2C%22Bloch%22%2C%22Bloch%22%2C%22Bloch%22%2C%22Bloch%22%2C%22Bloch%22%5D%2C%5B%22Amps9%22%5D%5D%2C%22gates%22%3A%5B%7B%22id%22%3A%22~590m%22%2C%22name%22%3A%22%E2%88%9A1%2F3%22%2C%22matrix%22%3A%22%7B%7B%E2%88%9A%E2%85%93%2C-%E2%88%9A%E2%85%94%7D%2C%7B%E2%88%9A%E2%85%94%2C%E2%88%9A%E2%85%93%7D%7D%22%7D%2C%7B%22id%22%3A%22~94a8%22%2C%22name%22%3A%22W3%22%2C%22circuit%22%3A%7B%22cols%22%3A%5B%5B%22X%22%5D%2C%5B%22%E2%80%A2%22%2C%22~590m%22%5D%2C%5B%22X%22%2C%22%E2%80%A2%22%2C%22H%22%5D%2C%5B1%2C%22X%22%2C%22%E2%80%A2%22%5D%5D%7D%7D%2C%7B%22id%22%3A%22~mv19%22%2C%22name%22%3A%22W2%22%2C%22circuit%22%3A%7B%22cols%22%3A%5B%5B%22X%22%5D%2C%5B%22%E2%80%A2%22%2C%22H%22%5D%2C%5B%22X%22%2C%22%E2%80%A2%22%5D%5D%7D%7D%2C%7B%22id%22%3A%22~bvbm%22%2C%22name%22%3A%22W2'%22%2C%22circuit%22%3A%7B%22cols%22%3A%5B%5B%22X%22%5D%2C%5B%22%E2%80%A2%22%2C1%2C%22H%22%5D%2C%5B%22X%22%2C1%2C%22%E2%80%A2%22%5D%5D%7D%7D%5D%7D
}{here (for $A_3$)}, or
\href{https://algassert.com/quirk#circuit=
}{here (for $A_4$)}.

\subsection{Discrete Portfolio Rebalancing}

As our third and last example we consider Discrete Portfolio Rebalancing, which was recently suggested as a first financial application for QAOA~\cite{hodson2019portfolio}.
In its simplest form, we are given a number of assets and 
a portfolio of short and long positions on these assets.
Periodically, such a portfolio has to be rebalanced in order to maintain in order to react to market and risk changes.
For all details regarding modeling and cost calculations we refer to the original article, here we focus on the the part pertaining to state preparation and mixer design.

\subsubsection{Problem encoding}

In the simplest form of Discrete Portfolio Rebalancing, we are
given $n$ assets. On each asset we may have a short position, a long position or neither, with the condition that we allocate
long and short positions for a net total (number of long positions minus number of short positions) of $d$ discrete lots. 

We encode a rebalanced portfolio with two length-$n$ bitstrings 
$\ell = \ell_{n}\ldots \ell_1$ and $s = s_n\ldots s_1$, 
where $\ell_i = 1$ stands for a long position on asset $i$, and 
$s_i=1$ for a short position, respectively. 
The cost function $C(\ell,s)$ depends on several aspects:
the risk-return of $(\ell,s)$, the trading costs of rebalancing the outdated portfolio into $(\ell,s)$, and a penalty term for the unwanted case of holding a long and a short position of the same asset ($\ell_i=s_i =1$).

A portfolio $(\ell,s)$ is said to belong to a \emph{band} $k$, if it has $k$ short positions, i.e. if $s$ has Hamming weight $HW(s)=k$.
There are $n-d+1$ different bands $k=0, \ldots, n-d$.
To achieve a given band $k$, we can choose $\binom{n}{k}$ different combinations of short positions 
and $\binom{n}{d+k}$ combinations of long positions for a total of $\binom{n}{d+k}\binom{n}{k}$ different portfolios.

\subsubsection{State preparation}
In order to create an equal superposition of all portfolios of net total $d$ lots, we precompute all the number of portfolios in each band $k$ as well as their prefix sums. For example, given 4 assets and a fixet net total of 2 discrete lots, there are
$\binom{4}{2}\binom{4}{0} = 6$ portfolios in band 0,
$\binom{4}{3}\binom{4}{1} = 16$ portfolios in band 1, and
$\binom{4}{4}\binom{4}{2} = 6$ portfolios in band 2.
The corresponding prefix sums are 6, 22, 28.

Next we give a brief overview of the Dicke state preparation unitaries in Theorem~\ref{thm:dicke}: 
In their most general form, they consist of unitaries $U_{n,n}$ of size $O(n^2)$ and depth $O(n)$, which map any computational basis states $\ket{1}^{\otimes k}\ket{0}^{\otimes n-k}$ to the equal superposition of all Hamming-weight-k states, namely the Dicke state:
\[  U_{n,n}\colon \ket{1}^{\otimes k}\ket{0}^{\otimes n-k} \mapsto \dicke{n}{k}.   \]
By linearity $U_{n,n}$ thus also maps superpositions of such input states to superpositions of Dicke states (symmetric states).

In particular, this means that to create a superposition of all portfolios, we can start with a correctly weighted superposition of bands, followed by one Dicke state unitary $U_{n,n}$ each for the $\ell$ and the $s$ register. 
Such weighted superpositions can be implemented with a stair of controlled $Y$-rotations with angles based on the precomputed band and prefix weights~\cite{dickestates}. 
For an outline see Figure~\ref{fig:portfolio} (left) and a
\href{https://algassert.com/quirk#circuit=%7B%22cols%22%3A%5B%5B1%2C1%2C1%2C%22~nuje%22%2C1%2C1%2C%22X%22%2C%22X%22%5D%2C%5B1%2C1%2C%22~hlnc%22%2C%22%E2%80%A2%22%5D%2C%5B1%2C1%2C1%2C%22%E2%80%A2%22%2C1%2C%22X%22%5D%2C%5B1%2C1%2C%22%E2%80%A2%22%2C1%2C%22X%22%5D%2C%5B1%2C1%2C%22Chance2%22%2C1%2C%22Chance2%22%5D%2C%5B%22~sj0p%22%2C1%2C1%2C1%2C%22~sj0p%22%5D%5D%2C%22gates%22%3A%5B%7B%22id%22%3A%22~vl6q%22%2C%22name%22%3A%22%E2%88%9A1%2F4%22%2C%22circuit%22%3A%7B%22cols%22%3A%5B%5B%7B%22id%22%3A%22Ryft%22%2C%22arg%22%3A%222acos(sqrt(1%2F4))%22%7D%5D%5D%7D%7D%2C%7B%22id%22%3A%22~tia4%22%2C%22name%22%3A%22%E2%88%9A2%2F4%22%2C%22circuit%22%3A%7B%22cols%22%3A%5B%5B%7B%22id%22%3A%22Ryft%22%2C%22arg%22%3A%222acos(sqrt(2%2F4))%22%7D%5D%5D%7D%7D%2C%7B%22id%22%3A%22~iel2%22%2C%22name%22%3A%22%E2%88%9A3%2F4%22%2C%22circuit%22%3A%7B%22cols%22%3A%5B%5B%7B%22id%22%3A%22Ryft%22%2C%22arg%22%3A%222acos(sqrt(3%2F4))%22%7D%5D%5D%7D%7D%2C%7B%22id%22%3A%22~fn9r%22%2C%22name%22%3A%22%E2%88%9A1%2F3%22%2C%22circuit%22%3A%7B%22cols%22%3A%5B%5B%7B%22id%22%3A%22Ryft%22%2C%22arg%22%3A%222acos(sqrt(1%2F3))%22%7D%5D%5D%7D%7D%2C%7B%22id%22%3A%22~gunh%22%2C%22name%22%3A%22%E2%88%9A2%2F3%22%2C%22circuit%22%3A%7B%22cols%22%3A%5B%5B%7B%22id%22%3A%22Ryft%22%2C%22arg%22%3A%222acos(sqrt(2%2F3))%22%7D%5D%5D%7D%7D%2C%7B%22id%22%3A%22~rtut%22%2C%22name%22%3A%22%E2%88%9A1%2F2%22%2C%22circuit%22%3A%7B%22cols%22%3A%5B%5B%7B%22id%22%3A%22Ryft%22%2C%22arg%22%3A%222acos(sqrt(1%2F2))%22%7D%5D%5D%7D%7D%2C%7B%22id%22%3A%22~nuje%22%2C%22name%22%3A%22%E2%88%9A6%2F28%22%2C%22circuit%22%3A%7B%22cols%22%3A%5B%5B%7B%22id%22%3A%22Ryft%22%2C%22arg%22%3A%222acos(sqrt(6%2F28))%22%7D%5D%5D%7D%7D%2C%7B%22id%22%3A%22~hlnc%22%2C%22name%22%3A%22%E2%88%9A16%2F22%22%2C%22circuit%22%3A%7B%22cols%22%3A%5B%5B%7B%22id%22%3A%22Ryft%22%2C%22arg%22%3A%222acos(sqrt(16%2F22))%22%7D%5D%5D%7D%7D%2C%7B%22id%22%3A%22~sj0p%22%2C%22name%22%3A%22U4%2C4%22%2C%22circuit%22%3A%7B%22cols%22%3A%5B%5B1%2C1%2C%22%E2%80%A2%22%2C%22X%22%5D%2C%5B1%2C1%2C%22~vl6q%22%2C%22%E2%80%A2%22%5D%2C%5B1%2C1%2C%22%E2%80%A2%22%2C%22X%22%5D%2C%5B1%2C%22%E2%80%A2%22%2C1%2C%22X%22%5D%2C%5B1%2C%22~tia4%22%2C%22%E2%80%A2%22%2C%22%E2%80%A2%22%5D%2C%5B1%2C%22%E2%80%A2%22%2C1%2C%22X%22%5D%2C%5B%22%E2%80%A2%22%2C1%2C1%2C%22X%22%5D%2C%5B%22~iel2%22%2C%22%E2%80%A2%22%2C1%2C%22%E2%80%A2%22%5D%2C%5B%22%E2%80%A2%22%2C1%2C1%2C%22X%22%5D%2C%5B1%2C%22%E2%80%A2%22%2C%22X%22%5D%2C%5B1%2C%22~fn9r%22%2C%22%E2%80%A2%22%5D%2C%5B1%2C%22%E2%80%A2%22%2C%22X%22%5D%2C%5B%22%E2%80%A2%22%2C1%2C%22X%22%5D%2C%5B%22~gunh%22%2C%22%E2%80%A2%22%2C%22%E2%80%A2%22%5D%2C%5B%22%E2%80%A2%22%2C1%2C%22X%22%5D%2C%5B%22%E2%80%A2%22%2C%22X%22%5D%2C%5B%22~rtut%22%2C%22%E2%80%A2%22%5D%2C%5B%22%E2%80%A2%22%2C%22X%22%5D%5D%7D%7D%5D%7D
}{interactive circuit},
where e.g. $\surd6/22$ is shorthand for a rotation of $R_y(2\cos^{-1}(\sqrt{6/22}))$.

\subsubsection{Improvement}
The original proposal of QAOA for Discrete Portfolion Rebalancing, displayed in Figure~\ref{fig:portfolio} (right), on the first glance shares some similarities:
First, Bell pairs between short and long positions give a superposition over bands $0,\ldots,n-d$. 
Secondly, Hamming-weight preserving mixing unitaries on registers $\ell$ and $s$ each 
are used to get superpositions over all portfolios in the respective bands.

However, there are two major differences: using Bell pairs to get a superposition over the bands results in a binomial weight distribution over the bands~\cite{hodson2019portfolio}, rather than a weight distribution that models the number of different portfolios in each band. 
Secondly -- and more importantly -- using Hamming-weight preserving mixing unitaries for the two registers prohibits a mixing between bands! Hence only the Grover mixing unitary can both restrict mixing to the feasible subspace \emph{and} provide transitions between all feasible states.

\subsubsection{Generalizations}
We believe this last example shows one strength of the GM-QAOA approach: Focusing on creating an equal superposition of all feasible states might be more amenable than the direct design of mixing unitaries. 
Furthermore, as a generalization one may drop the requirement of creating \emph{equal} superpositions of all feasible states 
and relax this to creating superpositions of all feasible states with a \emph{non-zero} amplitude for each feasible state.

\section{Conclusions}
We introduced the GM-QAOA framework for quantum algorithm optimization, which combines algorithmic principles from the Grover search algorithm and the Quantum Alternating Operator Ansatz. 
GM-QAOA mixer unitaries are straightforward to implement exactly -- and without any Hamiltonian simulation error -- compared to standard mixers from QAOA.

A review of applications for GM-QAOA has revealed multiple strengths:
Grover mixers can combine useful properties of existing mixers (such as implementability and good transition properties in the case of \textsc{Max-k-VertexCover}), 
reduce the circuit complexity compared to existing mixers (such as for \textsc{TSP}),
or even give rise to the first known mixing unitaries that both stay in the feasible subspace and provide transitions between all states therein (as in the case of Discrete Portfolio Optimization).

GM-QAOA requires an efficient method to create a superposition of all feasible solutions, without measurement.  
While not all optimization problems may have such circuits, a large number of interesting problems do, including the important Traveling Salesperson Problem and solution-size constraint maximization versions of many combinatorial problems, including $k$-Vertex Cover, $k$-Set Cover, etc. 

As an aside, we note that, as with many questions of quantum computational complexity nature, we do not have a complete understanding of which optimization problems fall into this class of allowing polynomial-time feasible solution superposition unitary operators, which is an interesting direction for future research. 
Other future directions include a numerical simulation assessment of our results and identifying ways to leverage the property of GM-QAOA of producing states where equal-quality solutions result in equal sampling probability.


\bibliographystyle{IEEEtran}
\bibliography{gq-bib.bib}

\begin{thebibliography}{10}
\providecommand{\url}[1]{#1}
\csname url@samestyle\endcsname
\providecommand{\newblock}{\relax}
\providecommand{\bibinfo}[2]{#2}
\providecommand{\BIBentrySTDinterwordspacing}{\spaceskip=0pt\relax}
\providecommand{\BIBentryALTinterwordstretchfactor}{4}
\providecommand{\BIBentryALTinterwordspacing}{\spaceskip=\fontdimen2\font plus
\BIBentryALTinterwordstretchfactor\fontdimen3\font minus
  \fontdimen4\font\relax}
\providecommand{\BIBforeignlanguage}[2]{{%
\expandafter\ifx\csname l@#1\endcsname\relax
\typeout{** WARNING: IEEEtran.bst: No hyphenation pattern has been}%
\typeout{** loaded for the language `#1'. Using the pattern for}%
\typeout{** the default language instead.}%
\else
\language=\csname l@#1\endcsname
\fi
#2}}
\providecommand{\BIBdecl}{\relax}
\BIBdecl

\bibitem{Hadfield2019}
S.~Hadfield, Z.~Wang, B.~O’Gorman, E.~G. Rieffel, D.~Venturelli, and
  R.~Biswas, ``{From the Quantum Approximate Optimization Algorithm to a
  Quantum Alternating Operator Ansatz},'' \emph{Algorithms}, vol.~12, no.~2,
  p.~34, 2019,
  \href{https://doi.org/doi:10.3390/a12020034}{doi:10.3390/a12020034}.

\bibitem{ambainis2019quantum}
A.~Ambainis, K.~Balodis, J.~Iraids, M.~Kokainis, K.~Prūsis, and J.~Vihrovs,
  ``{Quantum Speedups for Exponential-Time Dynamic Programming Algorithms},''
  in \emph{30th Annual ACM-SIAM Symposium on Discrete Algorithms, SODA'19},
  2019, pp. 1783--1793,
  \href{https://doi.org/10.1137/1.9781611975482.107}{doi:10.1137/1.9781611975482.107}.

\bibitem{Farhi2014}
E.~Farhi, J.~Goldstone, and S.~Gutmann, ``{A Quantum Approximate Optimization
  Algorithm},'' \emph{arXiv e-prints}, 2014,
  \href{https://arxiv.org/abs/1411.4028}{arXiv:1411.4028}.

\bibitem{nasa}
Z.~Wang, N.~C. Rubin, J.~M. Dominy, and E.~G. Rieffel, ``{$XY$ mixers:
  Analytical and numerical results for the quantum alternating operator
  ansatz},'' \emph{Physical Review A}, vol. 101, no.~1, p. 012320, 2020,
  \href{https://doi.org/10.1103/PhysRevA.101.012320}{doi:10.1103/PhysRevA.101.012320}.

\bibitem{eidenbenz2019quantum}
J.~Cook, S.~Eidenbenz, and A.~B{\"a}rtschi, ``{The Quantum Alternating Operator
  Ansatz on Max-k Vertex Cover},'' in \emph{IEEE International Conference on
  Quantum Computing \& Engineering, QCE'20}, 2020,
  \href{https://arxiv.org/abs/1910.13483}{arXiv:1910.13483}.

\bibitem{Farhi2015}
E.~Farhi, J.~Goldstone, and S.~Gutmann, ``{A Quantum Approximate Optimization
  Algorithm Applied to a Bounded Occurrence Constraint Problem},'' \emph{arXiv
  e-prints}, 2014, \href{https://arxiv.org/abs/1412.6062}{arXiv:1412.6062}.

\bibitem{grover_fixed_point}
L.~K. Grover, ``{Fixed-Point Quantum Search},'' \emph{Physical Review Letters},
  vol.~95, no.~15, p. 150501, 2005,
  \href{https://doi.org/10.1103/PhysRevLett.95.150501}{doi:10.1103/PhysRevLett.95.150501}.

\bibitem{yoder2014fixed}
T.~J. Yoder, G.~H. Low, and I.~L. Chuang, ``{Fixed-Point Quantum Search with an
  Optimal Number of Queries},'' \emph{Physical Review Letters}, vol. 113,
  no.~21, p. 210501, 2014,
  \href{https://doi.org/10.1103/PhysRevLett.113.210501}{doi:10.1103/PhysRevLett.113.210501}.

\bibitem{dickestates}
A.~B{\"a}rtschi and S.~Eidenbenz, ``{Deterministic Preparation of Dicke
  States},'' in \emph{22nd International Symposium on Fundamentals of
  Computation Theory, FCT'19}, 2019, pp. 126--139,
  \href{https://doi.org/10.1007/978-3-030-25027-0_9}{doi:10.1007/978-3-030-25027-0\_9}.

\bibitem{lloyd1996trotterization}
S.~Lloyd, ``{Universal Quantum Simulators},'' \emph{Science}, vol. 273, no.
  5278, pp. 1073--1078, 1996,
  \href{https://doi.org/10.1126/science.273.5278.1073}{doi:10.1126/science.273.5278.1073}.

\bibitem{berry2015simulating}
D.~W. Berry, A.~M. Childs, R.~Cleve, R.~Kothari, and R.~D. Somma, ``{Simulating
  Hamiltonian dynamics with a truncated Taylor series},'' \emph{Physical Review
  Letters}, vol. 114, no.~9, p. 090502, 2015,
  \href{https://doi.org/10.1103/PhysRevLett.114.090502}{doi:10.1103/PhysRevLett.114.090502}.

\bibitem{berry2015hamiltonian}
D.~W. Berry, A.~M. Childs, and R.~Kothari, ``Hamiltonian simulation with nearly
  optimal dependence on all parameters,'' in \emph{56th Annual IEEE Symposium
  on Foundations of Computer Science, FOCS'15}, 2015, pp. 792--809,
  \href{https://doi.org/10.1109/FOCS.2015.54}{doi:10.1109/FOCS.2015.54}.

\bibitem{grover1996}
L.~K. Grover, ``{A Fast Quantum Mechanical Algorithm for Database Search},'' in
  \emph{28th Annual ACM Symposium on Theory of Computing, STOC'96}, 1996, pp.
  212--219,
  \href{https://doi.org/10.1145/237814.237866}{doi:10.1145/237814.237866}.

\bibitem{Brassard1997}
G.~Brassard, ``{Searching a Quantum Phone Book},'' \emph{Science}, vol. 275,
  no. 5300, pp. 627--628, 1997,
  \href{https://doi.org/10.1126/science.275.5300.627}{doi:10.1126/science.275.5300.627}.

\bibitem{brassard2000}
G.~Brassard, P.~Hoyer, M.~Mosca, and A.~Tapp, ``{Quantum Amplitude
  Amplification and Estimation},'' \emph{arXiv e-prints}, 2000,
  \href{https://arxiv.org/abs/quant-ph/0005055}{arXiv:quant-ph/0005055}.

\bibitem{biamonte2018}
M.~E.~S. Morales, T.~Tlyachev, and J.~Biamonte, ``{Variational learning of
  Grover's quantum search algorithm},'' \emph{Physical Review A}, vol.~98,
  no.~6, p. 062333, 2018,
  \href{https://doi.org/10.1103/PhysRevA.98.062333}{doi:10.1103/PhysRevA.98.062333}.

\bibitem{biamonte2020}
V.~Akshay, H.~Philathong, M.~E.~S. Morales, and J.~D. Biamonte, ``{Reachability
  Deficits in Quantum Approximate Optimization},'' \emph{Physical Review
  Letters}, vol. 124, no.~9, p. 090504, 2020,
  \href{https://doi.org/10.1103/PhysRevLett.124.090504}{doi:10.1103/PhysRevLett.124.090504}.

\bibitem{sundar2019}
B.~Sundar, R.~Paredes, D.~T. Damanik, L.~Dueñas-Osorio, and K.~R.~A. Hazzard,
  ``{A quantum algorithm to count weighted ground states of classical spin
  Hamiltonians},'' 2019,
  \href{https://arxiv.org/abs/1908.01745}{arXiv:1908.01745}.

\bibitem{Matsuda2009}
Y.~Matsuda, H.~Nishimori, and H.~G. Katzgraber, ``Quantum annealing for
  problems with ground-state degeneracy,'' \emph{Journal of Physics: Conference
  Series}, vol. 143, p. 012003, 2009,
  \href{https://doi.org/10.1088/1742-6596/143/1/012003}{doi:10.1088/1742-6596/143/1/012003}.

\bibitem{OGorman2019}
B.~O'Gorman, W.~J. Huggins, E.~G. Rieffel, and K.~B. Whaley, ``Generalized swap
  networks for near-term quantum computing,'' \emph{arXiv e-prints}, 2019,
  \href{https://arxiv.org/abs/1905.05118}{arXiv:1905.05118}.

\bibitem{Gidney2015-Bootstrapping}
\BIBentryALTinterwordspacing
C.~Gidney, ``{U}sing quantum gates instead of ancilla bits,'' Jun 2015.
  [Online]. Available:
  \url{https://algassert.com/circuits/2015/06/22/Using-Quantum-Gates-instead-of-Ancilla-Bits.html}
\BIBentrySTDinterwordspacing

\bibitem{Gidney2015-Toffolis}
\BIBentryALTinterwordspacing
------, ``{C}onstructing large controlled nots,'' Jun 2015. [Online].
  Available:
  \url{https://algassert.com/circuits/2015/06/05/Constructing-Large-Controlled-Nots.html}
\BIBentrySTDinterwordspacing

\bibitem{Gidney2015-Increments}
\BIBentryALTinterwordspacing
------, ``{C}onstructing large increment gates,'' Jun 2015. [Online].
  Available:
  \url{https://algassert.com/circuits/2015/06/12/Constructing-Large-Increment-Gates.html}
\BIBentrySTDinterwordspacing

\bibitem{Rentergem2005}
Y.~V. Rentergem and A.~D. Vos, ``{Optimal Design of a Reversible Full Adder},''
  \emph{International Journal of Unconventional Computing}, vol.~1, pp.
  339--355, 2005.

\bibitem{givens-FFFT}
I.~D. Kivlichan, J.~McClean, N.~Wiebe, C.~Gidney, A.~Aspuru-Guzik, G.~K.-L.
  Chan, and R.~Babbush, ``{Quantum Simulation of Electronic Structure with
  Linear Depth and Connectivity},'' \emph{Physical Review Letters}, vol. 120,
  no.~11, p. 110501, 2018,
  \href{https://doi.org/10.1103/PhysRevLett.120.110501}{doi:10.1103/PhysRevLett.120.110501}.

\bibitem{fefferman2016}
B.~Fefferman and C.~Umans, ``{On the Power of Quantum Fourier Sampling},'' in
  \emph{11th Conference on the Theory of Quantum Computation, Communication and
  Cryptography, TQC'16}, ser. LIPICS, vol.~61, 2016, pp. 1:1--1:19,
  \href{https://doi.org/10.4230/LIPIcs.TQC.2016.1}{doi:10.4230/LIPIcs.TQC.2016.1}.

\bibitem{Bacon2006}
D.~Bacon, I.~L. Chuang, and A.~W. Harrow, ``{Efficient Quantum Circuits for
  Schur and Clebsch-Gordan Transforms},'' \emph{Physical Review Letters},
  vol.~97, no.~17, p. 170502, Oct 2006,
  \href{https://doi.org/10.1103/PhysRevLett.97.170502}{doi:10.1103/PhysRevLett.97.170502}.

\bibitem{QAP}
T.~C. Koopmans and M.~Beckmann, ``{Assignment Problems and the Location of
  Economic Activities},'' \emph{Econometrica}, vol.~25, no.~1, pp. 53--76,
  1957, \href{https://doi.org/10.2307/1907742}{doi:10.2307/1907742}.

\bibitem{MCES}
L.~Bahiense, G.~Manić, B.~Piva, and C.~C. de~Souza, ``{The maximum common edge
  subgraph problem: A polyhedral investigation},'' \emph{Discrete Applied
  Mathematics}, vol. 160, no.~18, pp. 2523--2541, 2012,
  \href{https://doi.org/10.1016/j.dam.2012.01.026}{10.1016/j.dam.2012.01.026}.

\bibitem{hodson2019portfolio}
M.~Hodson, B.~Ruck, H.~Ong, D.~Garvin, and S.~Dulman, ``{Portfolio rebalancing
  experiments using the Quantum Alternating Operator Ansatz},'' \emph{arXiv
  e-prints}, 2019, \href{https://arxiv.org/abs/1911.05296}{arXiv:1911.05296}.

\end{thebibliography}

\end{document}